\documentclass[10pt]{amsart}
\usepackage{graphicx,color}
\newtheorem{theorem}{Theorem}[section]
\newtheorem{lemma}[theorem]{Lemma}

\theoremstyle{definition}
\newtheorem{definition}[theorem]{Definition}

\theoremstyle{remark}

\numberwithin{equation}{section}

\usepackage{color}


\title[Inverse scattering on the quantum graph for graphene]{Inverse scattering on the quantum graph for graphene}

\author{Kazunori ANDO}
\address[K. Ando]{Department of Electrical and Electronic Engineering and Computer Science, Ehime University, Matsuyama, 790-8577, Japan}
\email{ando@cs.ehime-u.ac.jp}
\author{Hiroshi ISOZAKI}
\address[H. Isozaki]{Graduate School of Pure and Applied Sciences, Professor Emeritus,
University of Tsukuba, 
Tsukuba, 305-8571, Japan}
\email{isozakih@math.tsukuba.ac.jp}
\author{Evgeny KOROTYAEV}
\address[E. Korotyaev]{
Department of  Math. Analysis,
Saint-Petersburg State University,  Universitetskaya nab. 7/9, St. Petersburg, 199034, Russia, 
National Research University Higher School of Economics, St. Petersburg, Russia
}
\email{e.korotyaev@spbu.ru}
\author{Hisashi MORIOKA}
\address[H. Morioka]{Department of Electrical and Electronic Engineering and Computer Science, Ehime University, Matsuyama, 790-8577, Japan}
\email{morioka@cs.ehime-u.ac.jp}

\date{\today}

\begin{document}
\maketitle

\begin{abstract}
We consider the inverse scattering on the quantum graph associated with the hexagonal lattice. Assuming that the potentials on the edges are compactly supported and symmetric, we show that the S-matrix for all energies in any given open set in the continuous spectrum determines the potentials.
\end{abstract}

\section{Introduction}
In this paper, we are concerned with a family of one-dimensional Schr{\"o}dinger operators
$\displaystyle{- d^2/dz^2+ q_{\bf e}(z)}$
defined on the edges of the hexagonal lattice assuming the Kirchhoff condition on the vertices. Here, $z$ varies over the interval $(0,1)$ and ${\bf e} \in \mathcal E$, $\mathcal E$ being the set of all edges of the hexagonal lattice.  
The following assumptions are imposed on the potentials.

\medskip
\noindent
{\bf (Q-1)} \ \ \ \ \ \ {\it $q_{\bf e}(z)$ is real-valued, and $q_{\bf e} \in L^2(0,1)$}.

\medskip
\noindent
{\bf (Q-2)} \ \ \ \ \ \   {\it $q_{\bf e}(z) = 0$ on $(0,1)$ except for a finite  number of edges.}

\medskip
\noindent
{\bf (Q-3)} \ \ \ \ \ \  $q_{\bf e}(z) = q_{\bf e}(1-z)$ for $z \in (0,1)$.

\medskip
\noindent
Under these assumptions, the Schr{\"o}dinger operator 
$$
\widehat H_{\mathcal E} = \Big\{
- \frac{d^2}{dz^2} + q_{\bf e}(z)\, ; \, {\bf e} \in \mathcal E 
\Big\}
$$
is self-adjoint with  essential spectrum $\sigma_e(\widehat H_{\mathcal E}) = [0,\infty)$. 
  There exists a discrete (but infinte) subset $\mathcal T \subset {\bf R}$ such that $\sigma_e(\widehat H_{\mathcal E})\setminus \mathcal T$ is absolutely continuous.
  We can then define Heisenberg's S-matrix $S(\lambda)$ for $\lambda \in (0,\infty)\setminus\mathcal T$.   
 The following two theorems are the main purpose of this paper.

\begin{theorem}
\label{Maintheorem1}
Assume (Q-1), (Q-2) and (Q-3). Then,
given any open interval $I \subset (0,\infty)\setminus \mathcal T$, and the S-matrix $S(\lambda)$ for all $\lambda \in I$, one can uniquely reconstruct the potential $q_{\bf e}(z)$ for all ${\bf e} \in \mathcal E$.
\end{theorem}

Under our assumptions (Q-1), (Q-2), (Q-3), $S(\lambda)$ is 
meromorphic in the complex domain $\{{\rm Re}\,\lambda > 0\}$ with possible branch points at $\mathcal T$. Therefore, the assumption of Theorem \ref{Maintheorem1} is equivalent to the condition that we are given $S(\lambda)$ for all $\lambda \in (0,\infty)\setminus{\mathcal T}$.   
One can also deal with perturbation of periodic edge potentials.

\begin{theorem}
\label{Maintheorem2}
Assume (Q-1) and (Q-3). 
Assume that we are given a real $q_0(z) \in L^2(0,1)$ satisfying $q_0(z) = q_0(1-z)$ and $q_{\bf e}(z) = q_0(z)$ on $(0,1)$ except for a finite number of edges ${\bf e} \in \mathcal E$. Given an open interval $I \subset \sigma_e(\widehat H_{\mathcal E}) \setminus \mathcal T$ and the S-matrix $S(\lambda)$ for all $\lambda \in I$, one can uniquely reconstruct the potential $q_{\bf e}(z)$ for all edges ${\bf e}\in \mathcal E$.
\end{theorem}

It is well-known that there is a close connection between the Laplacian on the quantum graph and that on the associated vertex set (see e.g. \cite{Below85}, \cite{C97}, \cite{KuchPost},  \cite{Pank06}, \cite{BPG08}).
Therefore, the basic results on the spectral theory for the quantum graph are derived from those for the associated discrete Laplacian.
Sections \ref{BasicTheory}, \ref{ResolventEstimates} and \ref{SectionSmatrixtoDNmap} are devoted to this transfer. In particular, we show that the S-matrix for the whole quantum graph determines the Dirchlet-to-Neumann map in a finite region on which perturbations are confined (Theorem \ref{Aext-A=IBSigma-1ILemma}). The inverse problem is solved in \S \ref{SectionInversescattering}  by using the classical theorem of Borg \cite{Borg}.

The monographs \cite{CDGT88}, \cite{CDS95}, \cite{Ch97}, \cite{Post12}, \cite{BK13} are expositions of the graph spectra and related problems from algebraic, geometric,  physical  and functional analytic view points with slight different emphasis on them.  
The present situation of the study of quantum graph is well explained in the above mentioned books, especially in Chapter 7 of \cite{BK13} together with an abundance of references therein.  See also \cite{KoLo07}, \cite{Niikuni16} for more recent results. 
Plenty of deep results for the inverse problem on the quantum graph have been presented. See e.g. \cite{Below85}, 
\cite{KostSchr99}, \cite{Pivo00}, \cite{GutSmil01}, \cite{Bel04}, \cite{BroWei09}, \cite{Yurk05}, \cite{Kura08}, \cite{ChExTu11}, \cite{AvdBelMat11}, \cite{VisComMirSor11}, \cite{MochiTroosh12}, \cite{Yurko16}, \cite{Bonda20} and other papers cited in the above books. We must also mention \cite{Col98}, \cite{CurtMor00} on inverse problem for the planar discrete graph.
There are also many recent works on the spectral and (inverse) scattering theory for discrete Schr{\"o}dinger operators on perturbed periodic structures
\cite{KoSa15}, \cite{KoSa15b}, \cite{KoSa15c}, \cite{Nakamura14}, \cite{Tadano16},  \cite{ParRich18}, \cite{IsKo12}, \cite{Ando12}, \cite{IsMo}, \cite{AnIsoMo17}, \cite{AnIsoMo17(1)}, \cite{EKMN17}. In this paper, we use our previous results \cite{AnIsoMo17}, \cite{AnIsoMo17(1)} in the transfer from the discrete Laplacian to the quantum graph. 
Many parts of this paper, especially the part dealing with the the forward probem,  can be generalized to more general lattices, which can be seen in \cite{AnIsoMo(2020)}.

For a measure space $(M,d\mu)$, let $L^2(M;{\bf C}^n;d\mu)$ be the  space of the ${\bf C}^n$-valued functions on $M$. It is often denoted by $L^2(M;{\bf C}^n)$ or $L^2(M)$ when $n=1$. For Banach spaces $X$ and $Y$, let ${\bf B}(X;Y)$ be the set of all bounded operators from $X$ to $Y$, and ${\bf B}(X) = {\bf B}(X;X)$.

\section{Quantum graph}
\label{BasicTheory}

\subsection{Vertex Laplacian}
\label{ExampleHexalattice} 
We follow the standard formulation of metric graph (see e.g. \cite{KuchPost} or \cite{Pank06})\footnote{For  the figure  of hexagonal lattice, see  e.g. \cite{KuchPost} or \cite{AnIsoMo(2020)}.}.
In ${\bf R}^2$, let $p^{(1)} = (1,0)$, $p^{(2)}= (2,0)$, ${\bf v}_1 = \Big(\frac{3}{2},-\frac{\sqrt3}{2}\Big)$, $
{\bf v}_2 = \Big(\frac{3}{2},\frac{\sqrt3}{2}\Big)$, and ${\bf v}(n)= n_1{\bf v_1} + n_2{\bf v}_2$ for $n = (n_1,n_2)$.
We define the vertex set $\mathcal V$ by
\begin{equation}
\mathcal V = {\mathop\cup_{i=1}^2}\mathcal V_i, \quad \mathcal V_i = \{p^{(i)} + {\bf v}(n)\, ; \, n \in {\bf Z}^2\}.
\nonumber
\end{equation}
Let $I_2$ : $L^2_{loc}({\mathcal V};{\bf C}) \to L^2_{loc}({\bf Z}^2 ; {\bf C}^2)$ be defined by
\begin{equation}
I_2 : \widehat f(v) \rightarrow (I_2\widehat f)(n) =
\left(
\begin{array}{c}
\widehat f_1(n) \\
\widehat f_2(n)
\end{array}
\right) = 
\left(
\begin{array}{c}
\widehat f(p^{(1)} + {\bf v}(n))\\
\widehat f(p^{(2)} + {\bf v}(n))
\end{array}
\right).
\label{DefineI2}
\end{equation}
We often write $\widehat f(n)$ instead of $(I_2\widehat f)(n)$.
The Laplacian is defined by 
\begin{equation}
\big(\widehat{\Delta}_{\mathcal V}\widehat f\big)(n) = 
\frac{1}{3}
\left(
\begin{split}
\widehat f_2(n_1,n_2) + \widehat f_2(n_1-1,n_2) + \widehat f_2(n_1,n_2-1)\\
\widehat f_1(n_1,n_2) + \widehat f_1(n_1+1,n_2) + \widehat f_2(n_1,n_2+1)
\end{split}
\right),
\nonumber
\end{equation}
which is self-adjoint on $L^2(\mathcal V)$ equipped with the inner product\footnote{Note that the degree of each vertex in $\mathcal V$ is 3. }
\begin{equation}
(\widehat f,\widehat g) = 3\sum_{n \in {\bf Z}^2} \widehat f(n)\cdot \overline{\widehat g(n)}.
\nonumber
\end{equation}
Define the discrete Fourier transform  $\mathcal U_{\mathcal V} : L^2({\bf Z}^2;{\bf C}^2) \to L^2({\bf T}^2;{\bf C}^2)$ by
\begin{equation}
(\mathcal U_{\mathcal V}\widehat f)(x) = 
\sqrt{3}\, (2\pi)^{-1}\sum_{n\in{\bf Z}^2}e^{in\cdot x}\widehat f(n), \quad x \in {\bf T}^2 = {\bf R}^2/(2\pi{\bf Z})^2.
\nonumber
\end{equation}
Then on $L^2({\bf T}^2;{\bf C}^2)$, $\mathcal U_{\mathcal V}(- \Delta_{\mathcal V})\,\mathcal U_{\mathcal V}^{\ast}$ is the operator of multiplication by 
\begin{equation}
H_0(x) = -\frac{1}{3}\left(
\begin{array}{cc}
 0 & 1 + e^{ix_1} + e^{ix_2} \\
1 + e^{-ix_1} + e^{-ix_2} & 0
\end{array}
\right).
\label{DefineH0(x)}
\end{equation}
The edge set $\mathcal E$ consists of the segments ${\bf e}$ of length 1 with end points in $\mathcal V$, endowed with  arclength metric, as well as the identification with the interval $(0,1)$ : 
${\bf e} = \{(1-z){\bf e}(0) + z{\bf e}(1)\, ; \, 0 \leq z \leq 1\}$, where ${\bf e}(0), {\bf e}(1) \in \mathcal V$. We put
\begin{equation}
\mathcal E_v = \mathcal E_v(0) \cup \mathcal E_v(1),
\quad
\mathcal E_v(i) = \{{\bf e} \in \mathcal E\, ; \, {\bf e}(i) = v\}, \quad i = 0, 1.
\nonumber
\end{equation}
 For a function $\widehat f$ on an 
edge ${\bf e} \in \mathcal E_v$, we define $\widehat f'(v)$ to be the derivative at $v$ along ${\bf e}$.
A function $\widehat f = \{\widehat f_{\bf e}\}_{{\bf e} \in \mathcal E}$ defined on $\mathcal E$ is said to satisfy the {\it Kirchhoff condition} if

\smallskip
\noindent
{\bf (K-1)} \noindent $\  \widehat f$ {\it is continuous on $\mathcal E$}.  

\smallskip
\noindent
{\bf (K-2)} \ {\it  $\widehat f_{\bf e}  \in C^1([0,1])$ on each edge ${\bf e} \in \mathcal E$, and 
$\sum_{{\bf e}\in\mathcal E_v} \widehat f'_{\bf e} = 0
$ at any vertex $v \in \mathcal V$.}


\subsection{Edge Laplacian}
We consider  1-dimensional Schr{\"o}dinger operators  
\begin{equation}
h^{(0)}_{\bf e} = - d^2/dz^2, \quad h_{\bf e} = h^{(0)}_{\bf e} + q_{\bf e}(z)
\nonumber
\end{equation} 
on 
$L^2_{\bf e} = L^2(0,1)$.  We  define the Hilbert space $L^2(\mathcal E)$ of ${\bf C}$-valued $L^2$-functions $\widehat f = \big\{\widehat f_{\bf e}\big\}_{{\bf e} \in \mathcal E}$ on the edge set $\mathcal E$ : $L^2(\mathcal E) = {\mathop\oplus_{\bf e \in \mathcal E}}L^2_{\bf e}$
equipped with the inner product
\begin{equation}
(\widehat f,\widehat g)_{L^2({\mathcal E})}  = \sum_{{\bf e}\in {\mathcal E}}(\widehat f_{\bf e},\widehat g_{\bf e})_{L^2(0,1)}.
\label{L2mathcalE}
\end{equation}
Define the Hamiltonian
\begin{equation}
\widehat H_{\mathcal E} : \widehat u = \left\{{\widehat u}_{\bf e}\right\}_{{\bf e}\in \mathcal E} \to 
\left\{h_{{\bf e}}{\widehat u}_{\bf e}\right\}_{{\bf e}\in \mathcal E}
\label{S2HamitonianhatH}
\end{equation}
with domain $D(\widehat H_{\mathcal E})$ consisting of  $\widehat u_{\bf e} \in H^2(0,1)$ \footnote{This is the Sobolev space of order 2.} satisfying the Kirchhoff condition (K-1), (K-2)  
 and $\sum_{{\bf e}\in\mathcal E}\|h_{\bf e}{\widehat u}_{\bf e}\|_{L^2(0,1)}^2 < \infty$. 
Then, $\widehat H_{\mathcal E}$ 
is self-adjoint in $L^2(\mathcal E)$.
When $q_{\bf e} = 0$, $\widehat H_{\mathcal E}$ is denoted by $\widehat H_{\mathcal E}^{(0)}$ or 
$- \widehat\Delta_{\mathcal E}$, i.e. 
\begin{equation}
\big(- \widehat\Delta_{\mathcal E}\widehat u\big)_{\bf e}(z) = - \frac{d^2}{dz^2}\widehat u_{\bf e}(z), \quad {\bf e} \in \mathcal E.
\nonumber
\end{equation}
 We call it  {\it edge Laplacian}.
 Let $q_{\mathcal E}$ be the multiplication operator defined by
\begin{equation}
\big(q_{\mathcal E}\widehat f\big)_{\bf e}(z) = q_{\bf e}(z)\widehat f_{\bf e}(z), \quad {\bf e} \in \mathcal E.
\nonumber
\end{equation}
Then
$\widehat H_{\mathcal E} = \widehat H^{(0)}_{\mathcal E} + q_{\mathcal E}.$ We put
\begin{equation}
 \widehat R^{(0)}_{\mathcal E}(\lambda) = (\widehat H^{(0)}_{\mathcal E} - \lambda)^{-1}, \quad 
 \widehat R_{\mathcal E}(\lambda) = (\widehat H_{\mathcal E} - \lambda)^{-1}.
\label{S3ResolventsRandR0}
\end{equation}

Let 
$ - (d^2/dz^2)_D$
be the Laplacian on $(0,1)$ with boundary condition $u(0) = u(1)=0$.  
Let $\phi_{{\bf e}0}(z,\lambda), \phi_{{\bf e}1}(z,\lambda)$ be the solutions of 
\begin{equation}
\big(- d^2/dz^2 + q_{\bf e}(z) - \lambda\big)\phi = 0
\label{DiffEqontheedege}
\end{equation}
 with initial data
\begin{equation}
\left\{
\begin{split}
&\phi_{{\bf e}0}(0,\lambda) = 0, \\
& \phi'_{{\bf e}0}(0,\lambda)=1, 
\end{split}
\right.
\qquad
\left\{
\begin{split}
&\phi_{{\bf e}1}(1,\lambda) = 0, \\
& \phi'_{{\bf e}1}(1,\lambda)=-1. 
\end{split}
\right.
\nonumber
\end{equation}
In the following, we assume that 
\begin{equation}
	\lambda \not\in {\mathop\cup_{{\bf e} \in \mathcal E}}\sigma(-(d^2/dz^2)_D + q_{{\bf e}}(z)),
\nonumber
\end{equation}
which guarantees that $\phi_{{\bf e}0}(1,\lambda) \neq 0$ and $\phi_{{\bf e}1}(0,\lambda) \neq 0$. 
If $w, v \in \mathcal V$ are two end ponts of  an  edge ${\bf e} \in \mathcal E$, we define  $\psi_{wv}(z,\lambda)$  by
$$
\psi_{wv}(z,\lambda) =
\left\{
\begin{split}
& \phi_{{\bf e}0}(z,\lambda), \quad {\rm if} \quad {\bf e}(0)  = v, \\
& \phi_{{\bf e}1}(z,\lambda), \quad {\rm if} \quad {\bf e}(0) = w.
\end{split}
\right.
$$
Note that by the assumption (Q-3), we have  $\phi_{{\bf e}0}(z,\lambda) = \phi_{{\bf e}1}(1 - z,\lambda)$, hence
\begin{equation}
\psi_{wv}(1,\lambda) = \psi_{vw}(1,\lambda).
\nonumber
\end{equation}
\begin{definition}
We define the {\it reduced vertex Laplacian} $\widehat\Delta_{\mathcal V, \lambda}$ on $\mathcal V$ by
\begin{equation}
\big(\widehat\Delta_{\mathcal V, \lambda}\widehat u\big)(v) = \frac{1}{3}
\sum_{w \sim v}\frac{1}{\psi_{wv}(1,\lambda)}\widehat u(w), \quad v \in \mathcal V
\label{S3DefineVertexLaplacian}
\end{equation}
for  $\widehat u \in L^2_{loc}(\mathcal V)$, where $w \sim v$ means that there exists an  edge ${\bf e} \in \mathcal E$ such that $v, w$ are end points of ${\bf e}$.
We also define a scalar multiplication operator:
\begin{equation}
\big(\widehat Q_{\mathcal V,\lambda}\widehat u\big)(v) = \widehat Q_{v,\lambda}(v)\widehat u(v), 
\nonumber
\end{equation}
where
\begin{equation}
\widehat Q_{v,\lambda}(v) =  \frac{1}{3}\sum_{w \in \mathcal E_v}\frac{\psi'_{wv}(1,\lambda)}{\psi_{wv}(1,\lambda)}.
\label{S3DefinewidehatQVlambda}
\end{equation}
\end{definition}

The resolvent 
$r_{{\bf e}}(\lambda) = ( - (d^2/dz^2)_D + q_{\bf e}(z) - \lambda)^{-1}$ is written as 
\begin{equation}
(r_{\bf e}(\lambda)\widehat f)(v) = \int_0^z\frac{\phi_{{\bf e}1}(z,\lambda)\phi_{{\bf e}0}(t,\lambda)}{\phi_{{\bf e}0}(1,\lambda)}\widehat f(t)dt 
+ \int_z^1\frac{\phi_{{\bf e}0}(z,\lambda)\phi_{{\bf e}1}(t,\lambda)}{\phi_{{\bf e}1}(0,\lambda)}\widehat f(t)dt.
\nonumber
\end{equation}
We put
\begin{equation}
\begin{split}
\Phi_{{\bf e}0}(\lambda)\widehat f &= \frac{d}{dz}\big(r_{\bf e}(\lambda)\widehat f\big)\Big|_{z = 0} = \int_0^1\frac{\phi_{{\bf e}1}(t,\lambda)}{\phi_{{\bf e}1}(0,\lambda)}\widehat f(t)dt, \\
\Phi_{{\bf e}1}(\lambda)\widehat f &= -\frac{d}{dz}\big(r_{\bf e}(\lambda)\widehat f\big)\Big|_{z = 1} = \int_0^1\frac{\phi_{{\bf e}0}(t,\lambda)}{\phi_{{\bf e}0}(1,\lambda)}\widehat f(t)dt,
\end{split}
\nonumber
\end{equation}
and define an operator $\widehat T_{\mathcal V}(\lambda) : L^2_{loc}(\mathcal E) \to L^2_{loc}(\mathcal V)$ by
\begin{equation}
\begin{split}
\big(\widehat T_{{\mathcal V}}(\lambda)\widehat f\big)(v)  = \frac{1}{3}\Big(\sum_{{\bf e}\in\mathcal E_v(1)} \Phi_{{\bf e}1}(\lambda)\widehat f_{\bf e} + \sum_{{\bf e} \in \mathcal E_v(0)}\Phi_{{\bf e}0}(\lambda)\widehat f_{\bf e}\Big), \quad 
v \in \mathcal V.
\label{DefineFe(lambda)}
\end{split}
\end{equation}

Let $\widehat u = \{\widehat u_{\bf e}\}_{{\bf e}\in\mathcal E}$ be a solution to the equation $(\widehat H_{\mathcal{E}}-\lambda)\widehat u = \widehat f$. On each edge ${\bf e} \in \mathcal E$, it is written as 
\begin{equation}
\begin{split}
\widehat u_{{\bf e}}(z,\lambda) &= \Phi_{{\bf e}1}(\overline{\lambda})^{\ast}c_{{\bf e}}(1,\lambda) + 
\Phi_{{\bf e}0}(\overline{\lambda})^{\ast}c_{{\bf e}}(0,\lambda) + r_{{\bf e}}(\lambda)\widehat f_{{\bf e}} \\
&= \frac{\phi_{{\bf e}0}(t,\lambda)}{\phi_{{\bf e}0}(1,\lambda)}c_{\bf e}(1,\lambda) + \frac{\phi_{{\bf e}1}(t,\lambda)}{\phi_{{\bf e}1}(1,\lambda)}c_{\bf e}(0,\lambda) + r_{{\bf e}}(\lambda)\widehat f_{{\bf e}}
\end{split}
\label{widehatue(z,lambda)=PhiC+ref}
\end{equation}
with some constants $c_{{\bf e}}(0,\lambda)$, $c_{{\bf e}}(1,\lambda)$. Then, the condition (K-1) is satisfied if and only if for two edges ${\bf e}, {\bf e}' \in \mathcal E$ and  $p, q = 0, 1$, $c_{\bf e}(p,\lambda) = c_{{\bf e}'}(q,\lambda)$ if 
${\bf e}(p) = {\bf e}'(q)$.

\begin{lemma}
\label{Lemma3.2}
Let $\widehat u\big|_{\mathcal V}$ be the restriction of $\widehat u$ on $\mathcal V$. Then
the  condition (K-2) is rewritten as
\begin{equation}
\begin{split}
& \left(- \widehat\Delta_{\mathcal V,\lambda}+  \widehat{Q}_{\mathcal V,\lambda}\right)\widehat u\big|_{\mathcal V} = {\widehat T}_{\mathcal V}(\lambda)\widehat f.
\end{split}
\label{DeltaVlambdawidehataU+lambdaQVU=f}
\end{equation}
\end{lemma}
This lemma is well-known.
In fact, (K-2) is rewritten as
\begin{equation}
\begin{split}
&- \sum_{{\bf e}\in\mathcal E_v(0)}\frac{1}{\phi_{{\bf e}0}(1,\lambda)}c_{\bf e}(1,\lambda) - \sum_{{\bf e}\in\mathcal E_v(1)}\frac{1}{\phi_{{\bf e}1}(0,\lambda)}c_{\bf e}(0,\lambda)\\
& - \sum_{{\bf e}\in\mathcal E_v(0)}\frac{\phi'_{{\bf e}1}(0,\lambda)}{\phi_{{\bf e}1}(0,\lambda)}c_{\bf e}(0,\lambda) + \sum_{{\bf e}\in\mathcal E_v(1)}\frac{\phi'_{{\bf e}0}(1,\lambda)}{\phi_{{\bf e}0}(1,\lambda)}c_{\bf e}(1,\lambda)\\
&=  \sum_{{\bf e}\in\mathcal E_v(1)} \Phi_{{\bf e}1}(\lambda)\widehat f_{\bf e} + \sum_{{\bf e}\in\mathcal E_v(0)} \Phi_{{\bf e}0}(\lambda)\widehat f_{\bf e},
\end{split}
\nonumber
\end{equation}
which implies (\ref{DeltaVlambdawidehataU+lambdaQVU=f}). 
Therefore, $\widehat u\big|_{\mathcal V}$ should be written  as
\begin{equation}
\widehat u\big|_{\mathcal V} =  \left(- \widehat\Delta_{\mathcal V,\lambda}+  \widehat{Q}_{\mathcal V,\lambda}\right)^{-1}\widehat T_{\mathcal V}(\lambda)\widehat f.
\label{widehatu|mathcalVformal}
\end{equation}
Here, we must be careful about the operator $\big(- \widehat\Delta_{\mathcal V,\lambda}+  \widehat{Q}_{\mathcal V,\lambda}\big)^{-1}$.
For $\lambda \not\in {\bf R}$, the operator $-\widehat \Delta_{\mathcal V,\lambda} + \widehat Q_{\mathcal V,\lambda}$  has complex coefficients, hence is not self-adjoint. Therefore, the existence of its inverse is not obvious.  We discuss the validity of (\ref{widehatu|mathcalVformal})  in Subsection \ref{Limiting absorption principle}.  For the moment, we admit it as a formal formula. 

Noting that  $\widehat T_{\mathcal V}(\overline{\lambda})^{\ast} : L^2_{loc}(\mathcal V) \to L^2_{loc}(\mathcal E)$ is written as (see (\ref{widehatue(z,lambda)=PhiC+ref}))
\begin{equation}
(\widehat T_{\mathcal V}(\overline{\lambda})^{\ast}\widehat u)_{\bf e}(z) = 
\Phi_{{\bf e}1}(\overline{\lambda})^{\ast}\widehat u({\bf e}(1)) + 
\Phi_{{\bf e}0}(\overline{\lambda})^{\ast}\widehat u({\bf e}(0)),
\nonumber
\end{equation}
we have the following lemma by (\ref{DeltaVlambdawidehataU+lambdaQVU=f}).  Let $r_{\mathcal E}(\lambda) \in 
{\bf B}(L^2(\mathcal E))$ be defined by
\begin{equation}
r_{\mathcal E}(\lambda)\widehat f = r_{e}(\lambda)\widehat f_{\bf e}, \quad {\rm on} \quad {\bf e}.
\nonumber
\end{equation}
\begin{lemma}
\label{ResolventofwidehatHE}
The resolvent of $\widehat H_{\mathcal E}$ is  written as
\begin{equation}
\widehat R_{\mathcal E}(\lambda)= 
\widehat T_{\mathcal V}(\overline\lambda)^{\ast}\left(- \widehat\Delta_{\mathcal V,\lambda} + \widehat Q_{\mathcal V,\lambda}\right)^{-1}\widehat T_{\mathcal V}(\lambda) + r_{\mathcal E}(\lambda).
\label{Resolventformalformula1}
\end{equation}
\end{lemma}

For the unperturbed case $\widehat q_{\mathcal E} = 0$, we put the superscript $^{(0)}$ for every term. Then, we have
\begin{equation}
\phi_{{\bf e}0}^{(0)}(z) = \frac{\sin\sqrt{\lambda}z}{\sqrt{\lambda}}, \quad \phi_{{\bf e}1}^{(0)}(z) = \frac{\sin\sqrt{\lambda}(1-z)}{\sqrt{\lambda}}.
\label{phie10unperturbed}
\end{equation}
Therefore by (\ref{S3DefineVertexLaplacian}) and (\ref{S3DefinewidehatQVlambda}),  
\begin{equation}
\Big(\widehat\Delta^{(0)}_{\mathcal V,\lambda}\widehat u\Big)(v) = \frac{\sqrt{\lambda}}{\sin\sqrt{\lambda}}\frac{1}{3}\sum_{w\in \mathcal E_v}\widehat u(w) = \frac{\sqrt{\lambda}}{\sin\sqrt{\lambda}}\big(\widehat\Delta_{\mathcal V}\widehat u\big)(v),
\label{DefineDeltaVlambda(0)}
\end{equation}
\begin{equation}
\widehat Q^{(0)}_{v,\lambda} = \frac{\sqrt{\lambda}}{\sin{\sqrt{\lambda}}}\cos\sqrt{\lambda}.
\label{DefinewidehatQ0}
\end{equation}
Lemma \ref{ResolventofwidehatHE} then implies the following formula
\begin{equation}
\widehat R_{\mathcal E}^{(0)}(\lambda) 
= \widehat T^{(0)}_{\mathcal V}(\overline\lambda)^{\ast}\frac{\sin{\sqrt{\lambda}}}{\sqrt{\lambda}}\big(- \widehat\Delta_{\mathcal V} + \cos\sqrt{\lambda}\big)^{-1}\widehat T^{(0)}_{\mathcal V}(\lambda) + r^{(0)}_{\mathcal E}(\lambda).
\label{Resolventformula2}
\end{equation}


\section{Resolvent estimates}
\label{ResolventEstimates}
\subsection{Limiting absorption principle}
\label{Limiting absorption principle}
 In our previous work \cite{AnIsoMo17}, we proved  resolvent estimates of vertex Laplacian $- \Delta_{\mathcal V}$ in weighted $L^2$ spaces or Besov spaces of ${\bf C}$-valued functions. By virtue of the formulas (\ref{Resolventformalformula1}) and (\ref{Resolventformula2}),  the resolvent estimates of  edge Laplacian are derived from those of vertex Laplacian using the space of $L^2({\bf e})$-valued functions on the edge set $\mathcal E$ defined as follows.

For ${\bf e} \in \mathcal E$, we put
\begin{equation}
c({\bf e}) = \frac{1}{2}\left({\bf e}(0) + {\bf e}(1)\right).
\nonumber
\end{equation}
Letting $r_{-1}=0$, $r_j = 2^j$ ($j \geq 0$),  we define
\begin{eqnarray}
\nonumber
\widehat{\mathcal B}({\mathcal E}) \ni \widehat f 
&\Longleftrightarrow&
\|\widehat f\|_{\widehat{\mathcal B}(\mathcal E)} 
= \sum_{j=0}^{\infty}r_j^{1/2}
\Big(\sum_{r_{j-1}\leq |c({\bf e})| < r_j}\|\widehat f_{{\bf e}}\|^2_{L^2(0,1)}\Big)^{1/2} < \infty,
\nonumber
\\
\widehat{\mathcal B}^{\ast}({\mathcal E}) \ni \widehat f 
&\Longleftrightarrow&
\|\widehat f\|^2_{\widehat{\mathcal B}^{\ast}(\mathcal E)} = \sup_{R>1}\frac{1}{R}\sum_{|c({\bf e})| < R}\|\widehat f_{{\bf e}}\|_{L^2(0,1)}^2  < \infty,
\nonumber
\\
\widehat{\mathcal B}^{\ast}_{0}(\mathcal E)\ni \widehat f 
&\Longleftrightarrow&
\lim_{R\to\infty}\frac{1}{R}\sum_{|c({\bf e})| < R}\|\widehat f_{{\bf e}}\|_{L^2(0,1)}^2  =0.
\nonumber
\\
{\widehat L}^{\; 2,s}(\mathcal E) \ni \widehat f
&\Longleftrightarrow& 
\sum_{{\bf e} \in {\mathcal E}}
(1 + |c({\bf e})|^2)^{s}\|\widehat f_{{\bf e}}\|_{L^2(0,1)}^2  < \infty,\quad s \in {\bf R}.
\nonumber
\end{eqnarray}
 The function spaces $\widehat{\mathcal B}(\mathcal V)$ etc on the vertex set $\mathcal V$ are defined similarly for ${\bf C}$-valued functions with $c({\bf e})$ and $\widehat f_{\bf e}$ replaced by $v$ and $\widehat f(v)$, $v \in \mathcal V$, respectively.

Taking account of (\ref{Resolventformula2}), we define the characteristic surface of 
$- \Delta^{(0)}_{\mathcal V, \lambda}$ by 
$$
M_{\lambda} = \{x \in {\bf T}^2\, ; \, \det(H_0(x) + \cos\sqrt{\lambda}) = 0\}.
$$
Lemma 3.3 of \cite{AnIsoMo17} implies that $M_{\lambda}$ is smooth if $\cos\sqrt{\lambda} \neq 0, \pm 1/2, \pm 1,  \lambda \in {\bf R}$. Note that
$$
\sigma(-(d^2/dz^2)_D) = \{(\pi n)^2\, ;\, \, n \in {\bf Z}\} = 
\{\lambda\, ; \,\cos\sqrt{\lambda} = \pm 1\}.
$$
We put
\begin{equation}
\mathcal T^{(0)} = \{\lambda\, ; \, \cos\sqrt{\lambda} = 0, \pm 1/3, \pm 1\},
\nonumber
\end{equation}
\begin{equation}
\mathcal T = \mathcal T^{(0)} \cup\big(\cup_{{\bf e} \in \mathcal E}\sigma(- (d^2/dz^2)_D + q_{{\bf e}}(z))\big).
\label{DefinemathcalT}
\end{equation}

Let us return to the problem for $\big(- \widehat\Delta_{\mathcal V,\lambda}+  \widehat{Q}_{\mathcal V,\lambda}\big)^{-1}$ we have encountered in \S \ref{BasicTheory}. 
First we consider the case $q_{\mathcal E} = 0$. For $\lambda \in (0,\infty)\setminus{\mathcal T}^{(0)}$, 
arguing as in the proof of Theorem 7.7 in \cite{AnIsoMo17}, one can prove the uniform boundedness of $\big(- \widehat\Delta_{\mathcal V,\lambda}+  \widehat{Q}_{\mathcal V,\lambda \pm i\epsilon} \big)^{-1}$ and the existence of strong limit in ${\bf B}(L^{2,s}(\mathcal V) ; L^{2,-s}(\mathcal V))$, $s > 1/2$, and weak $\ast$-limit in ${\bf B}(\mathcal B(\mathcal V) ; {\mathcal B}^{\ast}(\mathcal V))$ of $\big(- \widehat\Delta_{\mathcal V,\lambda}+  \widehat{Q}_{\mathcal V,\lambda \pm i0} \big)^{-1}$.    
The arguments in \S \ref{BasicTheory} are then justified if we consider
all operators in $\mathcal B(\mathcal E)$ or $\mathcal B^{\ast}(\mathcal E)$. The limiting absorption principle is then extended to the edge Laplacian in the following way.
\begin{theorem}
\label{LAPwholespace}
(1) For any compact interval  $I$ in $(0,\infty)\setminus\mathcal T$, 
 there exists a constant $C > 0$ such that for any $\lambda \in I$ and $\epsilon > 0$
\begin{equation}
\|( \widehat H_{\mathcal E} - \lambda \mp i\epsilon)^{-1}\|_{{\bf B}(\mathcal B(\mathcal E);{\mathcal B}^{\ast}(\mathcal E))} \leq C.
\end{equation}

\noindent
(2) For any $\lambda \in (0,\infty)\setminus {\mathcal T}$ and $s > 1/2$, there exists a strong limit 
\begin{equation}
{\mathop{\rm s-lim}_{\epsilon\downarrow 0}}( \widehat H_{\mathcal E} - \lambda \mp i\epsilon)^{-1} := (\widehat H_{\mathcal E} - \lambda \mp i0)^{-1} \in 
{\bf B}\big({\widehat L}^{2,s}(\mathcal E);{\widehat L}^{2,-s}(\mathcal E)\big),
\end{equation}
and for any $\widehat f \in \widehat L^{2,s}(\mathcal E)$, $(\widehat H_{\mathcal E} - \lambda \mp i0)^{-1}\widehat f$ is an $\widehat L^{2,-s}(\mathcal E)$-valued strongly continuous function of $\lambda$.\\
\noindent
(3) For any $\widehat f, \widehat g \in \widehat{\mathcal B}(\mathcal E)$ and $\lambda \in (0,\infty)\setminus{\mathcal T}$, there exists a limit
\begin{equation}
{\mathop{\lim}_{\epsilon\downarrow 0}}\big((\widehat H_{\mathcal E} - \lambda \mp i\epsilon)^{-1}\widehat f,\widehat g\big) := \big(( \widehat H_{\mathcal E} - \lambda \mp i0)^{-1}\widehat f,\widehat g\big),
\end{equation}
and $ \big((\widehat H_{\mathcal E} - \lambda \mp i0)^{-1}\widehat f,\widehat g\big)$ is a continuous function of $\lambda$.
\end{theorem}
\subsection{Analytic continuation of the resolvent}
It is well-known that for the Schr{\"o}dinger operator $- \Delta + V(x)$ in ${\bf R}^d$, where $V(x)$ has compact support, the boundary value of the resolvent $( - \Delta + V(x) - \lambda - i0)^{-1}$ has a meromorphic continuation into the lower half plane $\{{{\rm Re}\, \lambda > 0, \ \rm Im}\, \lambda < 0\}$ as an operator from the space of compactly supported $L^2({\bf R}^d)$ functions to $L^2_{loc}({\bf R}^d)$. This is proven by considering the free case, i.e.  the operator
$$
\int_{{\bf R}^d}\frac{e^{ix\cdot\xi}\widetilde f(\xi)}{|\xi|^2 - \zeta}d\xi = 
\int_0^{\infty}\frac{\int_{S^{d-1}}e^{ir\omega\cdot x}\widetilde f(r\omega)d\omega}{r^2 - \zeta}r^{d-1}dr 
$$
($\widetilde f(\xi)$ being the Fourier transfrom of $f$)
for ${\rm Im}\,\zeta > 0$, deforming the path of integration into the lower half-plane, and then applying the perturbation theory. This method also works for the discrete case, and one can show that the resolvents of the vertex Hamiltonian and the edge Hamiltonian defined for $\{{{\rm Re}\,\lambda > 0, \ \rm Im}\, \lambda > 0\}$ can be continued meromorphically into the lower half-plane $\{{{\rm Re}\,\lambda > 0, \ \rm Im}\, \lambda < 0\}$ with possible branch points on $\mathcal T$, when the perturbation is compactly supported. 
\subsection{Spectral representation}
We can then construct the spectral representation of the edge Laplacian. 
Letting $P_{\mathcal V,j}(x)$ be the eigenprojection associated with the eigenvalue $\lambda_j(x)$ of $H_0(x)$, 
we put
\begin{equation}
\begin{split}
D^{(0)}(\lambda \pm i0)  &= \frac{\sin\sqrt{\lambda}}{\sqrt{\lambda}}\, \mathcal U_{\mathcal V}I_2\big(- \widehat\Delta_{\mathcal V} + \cos\sqrt{\lambda \pm i0}\big)^{-1}I_2^{\ast}{\mathcal U_{\mathcal V}}^{\ast}\\
&= \frac{\sin\sqrt{\lambda}}{\sqrt{\lambda}}\sum_{j=1}^2\frac{1}{\lambda_j(x) + \cos\sqrt{\lambda} \mp i\sigma(\lambda)0}P_{\mathcal V,j}(x),
\end{split}
\label{S5D0lambdaandresolvent}
\end{equation}
where $\sigma(\lambda) = 1$ if $\lambda > 0, \sin\sqrt{\lambda} > 0$, $\sigma(\lambda) = -1$ if $\lambda > 0, \sin\sqrt{\lambda} < 0$.
We also put
\begin{equation}
\Phi^{(0)}(\lambda)  = \mathcal U_{\mathcal V}I_2\widehat T^{(0)}_{\mathcal V}(\lambda).
\end{equation}
By (\ref{Resolventformula2}),   $\widehat R^{(0)}_{\mathcal E}(\lambda \pm i0)$ is rewritten as
\begin{equation}
\widehat R^{(0)}_{\mathcal E}(\lambda \pm i0) = 
\Phi^{(0)}(\lambda)^{\ast}D^{(0)}(\lambda \pm i0)\Phi^{(0)}(\lambda) + r^{(0)}_{\mathcal E}(\lambda).
\label{R0Elambda=Phi0AlambdaPhi0} 
\end{equation}

 We put
\begin{equation}
 M_{\lambda} = \cup_{j=1}^2M_{\lambda,j},
\quad 
M_{\lambda,j} = \{x \in {\bf T}^d\, ; \, \lambda_j(x) + \cos\sqrt{\lambda} = 0\}, 
\nonumber
\end{equation}
\begin{equation}
(\varphi,\psi)_{\lambda,j} = \int_{M_{\lambda,j}}\varphi(x)\overline{\psi(x)}dS_j, \quad
dS_j = \frac{|\sin \sqrt{\lambda}|}{\sqrt{\lambda}} \frac{dM_{\lambda,j}}{|\nabla_x\lambda_j(x)|},
\nonumber
\end{equation}
where $dM_{\lambda,j}$ is the induced measure on $M_{\lambda,j}$. 
For $\widehat f \in \mathcal B(\mathcal E)$, we define $\widehat{\mathcal F}_j^{(0)}(\lambda)\widehat f$ by
$$
\widehat{\mathcal F}_j^{(0)}(\lambda)\widehat f = \big(P_{\mathcal V,j}(x)\Phi^{(0)}(\lambda)\widehat f\big)\big|_{M_{\lambda}},
$$
i.e. the restriction of $P_{\mathcal V,j}(x)\Phi^{(0)}(\lambda)\widehat f$ to $M_{\lambda}$,
and
\begin{equation}
\widehat{\mathcal F}^{(0)}(\lambda) = \big(\widehat{\mathcal F}^{(0)}_{1}(\lambda),\widehat{\mathcal F}^{(0)}_{2}(\lambda)\big), 
\label{S5DefineFj(0)lambda}
\end{equation} 
\begin{equation}
{\bf h}_{\lambda} = L^2\big(M_{\lambda}\big) = {\mathop\oplus_{j=1}^2}L^2\big(M_{\lambda,j};dS_j\big),
\nonumber
\end{equation}
\begin{equation}
{\bf H} = L^2\big((0,\infty),{\bf h}_{\lambda};d\lambda\big).
\nonumber
\end{equation}
Noting that $\widehat{\mathcal F}^{(0)}(\lambda) \in {\bf B}(\mathcal B(\mathcal E);{\bf h}_{\lambda})$, the  spectral representation associated with  $\widehat H_{\mathcal E}$ is constructed by the perturbation method. Define $\widehat{\mathcal F}^{(\pm)}(\lambda)$ by
\begin{equation}
\widehat{\mathcal F}^{(\pm)}(\lambda) = \widehat{\mathcal F}^{(0)}(\lambda)\Big(1 - q_{\mathcal E}\widehat R_{\mathcal E}(\lambda \pm i0)\Big) \in 
{\bf B}(\mathcal B(\mathcal E)\, ; \,{\bf h}_{\lambda}).
\label{DefineFpmlambdainmathcalE}
\end{equation}
Then we have
\begin{equation}
\frac{1}{2\pi i}\Big(\big(\widehat R_{\mathcal E}(\lambda + i0)-\widehat R_{\mathcal E}(\lambda - i0)\big)\widehat f,\widehat g\Big) 
= (\widehat{\mathcal F}^{(\pm)}(\lambda)\widehat f,\widehat{\mathcal F}^{(\pm)}(\lambda)\widehat g)_{{\bf h}_{\lambda}}.
\label{RmathcalEParseval}
\end{equation}
We can  prove (\ref{RmathcalEParseval}) first for $\widehat H^{(0)}_{\mathcal E}$ by (\ref{R0Elambda=Phi0AlambdaPhi0}), and then for $\widehat H_{\mathcal E}$ 
by using the resolvent equation (see Lemma 7.8 in \cite{AnIsoMo17}). 
\begin{theorem}
\label{EigenfunctionExpansionWholespace}
(1) The operator $\widehat{\mathcal F}^{(\pm)}$ is uniquely extended to a partial isometry with initial set ${\mathcal H}_{ac}(\widehat H_{\mathcal E})$ and final set ${\bf H}$ annihilating ${\mathcal H}_p(\widehat H_{\mathcal E})$, the point spectral subspace for $\widehat H_{\mathcal E}$. 

\noindent
(2) It diagonalizes $\widehat H_{\mathcal E}$:
$$
\big(\widehat{\mathcal F}^{(\pm)}\widehat H_{\mathcal E}\widehat f\big)(\lambda) = 
\lambda\big(\widehat{\mathcal F}^{(\pm)}\widehat f\big)(\lambda), \quad \forall \widehat f \in D(\widehat H_{\mathcal E}).
$$
(3) The adjoint operator $\widehat{\mathcal F}^{(\pm)}(\lambda)^{\ast} \in {\bf B}({\bf h}_{\lambda};\mathcal B^{\ast}(\mathcal E))$ is an eigenoperator in the sense that
$$
(\widehat H_{\mathcal E} - \lambda)\widehat{\mathcal F}^{(\pm)}(\lambda)^{\ast}\phi = 0, \quad \forall \phi \in {\bf h}_{\lambda}.
$$
(4) For $\widehat f \in \mathcal H_{ac}(\widehat H_{\mathcal E})$, the inversion formula holds:
$$
\widehat f = \int_0^{\infty}\widehat{\mathcal F}^{(\pm)}(\lambda)^{\ast}\big(\widehat{\mathcal F}^{(\pm)}\widehat f\big)(\lambda)d\lambda.
$$
\end{theorem}

The crucial step for the inverse scattering procedure is Theorem \ref{HelmhlotzEqWholespace} below, which can be proven by the same argument as in \cite{AnIsoMo17}. We do not repeat the whole procedure, but explain important intermediate steps.
Let us prepare a lemma. 
\begin{lemma}
\label{S2B0astEandB0astVLemma}
 For a solution $\widehat u$ of the equation 
 $(\widehat H_{\mathcal E} - \lambda)\widehat u = \widehat f$
 satisfying the Kirchhoff condition, 
we have the inequality
 \begin{equation}
 C_{\lambda}^{-1}\|\widehat u\|_{{\mathcal B}^{\ast}(\mathcal E)} \leq 
 \|\widehat u\big|_{\mathcal V}\|_{\mathcal B^{\ast}(\mathcal V)} \leq C_{\lambda}\|\widehat u\|_{\mathcal B^{\ast}(\mathcal E)},
 \nonumber
 \end{equation}
 and the equivalence
 \begin{equation}
 \widehat u \in \mathcal B^{\ast}_0({\mathcal E}) \Longleftrightarrow \widehat u\big|_{\mathcal V} \in \mathcal B^{\ast}_0(\mathcal V).
 \nonumber
 \end{equation}
 \end{lemma}
 
 \begin{proof}
Note that $\widehat u_{\bf e}(z)$ is written as in (\ref{widehatue(z,lambda)=PhiC+ref}).  Since $\phi_{{\bf e}0}(t,\lambda)$ and $\phi_{{\bf e}1}(t,\lambda)$ are linearly independent, there exists a constant $C_{\lambda} > 0$ independent of ${\bf e}$ such that 
\begin{equation}
C_{\lambda}^{-1}\|\widehat u_{\bf e}\|_{L^2_{\bf e}} \leq 
|c_{\bf e}(0)| + |c_{\bf e}(1)| \leq C_{\lambda}\|\widehat u_{\bf e}\|_{L^2_{\bf e}}.
\nonumber
\end{equation}
The lemma then  follows from this inequality. 
\end{proof}

\noindent
(I) {\it Rellich type theorem.}
 We define  exterior and interior domains $\mathcal E_{ext,R}$ and $\mathcal E_{int,R}$ in $\mathcal E$ by
\begin{equation}
\mathcal E_{ext,R} \ni {\bf e} \Longleftrightarrow 
|c({\bf e})| \geq R,
\quad
\mathcal E_{int,R} \ni {\bf e} \Longleftrightarrow 
|c({\bf e})| < R.
\nonumber
\end{equation}
\begin{theorem}
\label{RelichTypeTheorem}
Let $\lambda \in (0,\infty)\setminus \mathcal T^{(0)}$, and suppose  $\widehat u \in \widehat {\mathcal B}^{\ast}_0(\mathcal E)$ satisfies 
$
\widehat H^{(0)}_{\mathcal E}\widehat u = \lambda\widehat u \ {\rm in} \ \mathcal E_{ext,R},
$
and the Kirchhoff condition 
for some $R > 0$. Then $\widehat u = 0$ on $\mathcal E_{ext,R_1}$ for some $R_1 > 0$.
 \end{theorem}

\begin{proof}
By Lemma \ref{S2B0astEandB0astVLemma}, 
$\widehat u\big|_{\mathcal V} \in {\mathcal B}_0^{\ast}(\mathcal V)$. Since $(- \Delta_{\mathcal V} + \cos\sqrt{\lambda})\widehat u\big|_{\mathcal V} = 0$ near infinity, 
by Theorem 5.1 in \cite{AnIsoMo17}, $\widehat u\big|_{\mathcal V} = 0$ near infinity. This proves Theorem \ref{RelichTypeTheorem}. 
\end{proof}

We say that the operator $\widehat H_{\mathcal E}-\lambda$ has the {\it unique continuation property} on $\mathcal E$ when the following assertion holds: If $\widehat u$ satisfies $(\widehat H_{\mathcal E} - \lambda)\widehat u = 0$ on $\mathcal E$ and $\widehat u = 0$ on $\mathcal E_{ext,R}$ for some $R > 0$, then $\widehat u = 0$ on $\mathcal E$.
The following lemma can be checked easily.
\begin{lemma}
For the hexagonal lattice in ${\bf R}^2$, the unique continuation property holds.
\end{lemma}

\noindent
(II) {\it Radiation condition}. 
The radiation condition for the vertex Laplacian was introduced in \cite{AnIsoMo17} for the distinction between $(- \Delta_{\mathcal V} - \lambda - i0)^{-1}$ and 
$(- \Delta_{\mathcal V} - \lambda +  i0)^{-1}$. Hence it is  extended to the edge Laplacian. 
Note that for the edge Laplacian, one must replace $\lambda$ in the definition (6.2) of \cite{AnIsoMo17} by $\cos\sqrt{\lambda}$. 
See \cite{AnIsoMo(2020)} for details.

\begin{theorem}
\label{LemmaRadCondUnique}
Let $\lambda \in (0,\infty)\setminus{\mathcal T}$ and $\widehat f \in \mathcal B(\mathcal E)$. \\
\noindent
(1) The solution $\widehat u \in \widehat{\mathcal B}^{\ast}({\mathcal E})$ of the equation $(- \widehat\Delta_{\mathcal E} + q_{\mathcal E} - \lambda)\widehat u = \widehat f$ satisfying the outgoing or incoming radiation condition is unique. \\
\noindent
(2) $( \widehat H_{\mathcal E} - \lambda - i0)^{-1}\widehat f$ satisfies the outgoing radiation condition, and $(\widehat H_{\mathcal E} - \lambda + i0)^{-1}\widehat f$ satisfies the incoming radiation condition.
\end{theorem}

\noindent
(III) {\it Singularity expansion}. Asymptotic behavior at infinity of the resolvent is closely related to the far-field behavior of the scattering waves. For the case of scattering on perturbed lattices,
instead of observing the asymptotic expansion  of $(\widehat H_{\mathcal E} - \lambda \mp i0)^{-1}$ at infinity of the edge space $\mathcal E$, it is more convenient to consider the  singularities of its Fourier transform in $\mathcal B^{\ast}$. 
For $f, g \in \mathcal B^{\ast}({\mathcal E})$, we use the notation $f \simeq g$ in the following sense:
$$
f \simeq g \Longleftrightarrow 
f - g \in \mathcal B^{\ast}_0({\mathcal E}).
$$
We  use the same notation $\simeq$ for ${\mathcal B}^{\ast}(\mathcal V)$.

Since $r_{\mathcal E}(\lambda)$ is bounded on $L^2(\mathcal E)$, we have $r_{\mathcal E}(\lambda)\widehat f \simeq 0$ for any $\widehat f \in L^2(\mathcal E)$. 
Then, by (\ref{widehatue(z,lambda)=PhiC+ref}) and (\ref{widehatu|mathcalVformal}), the singularities of $(\widehat H_{\mathcal E} - \lambda \mp i0)^{-1}$ appear  from 
$\left(- \widehat\Delta_{\mathcal V,\lambda}+  \widehat{Q}_{\mathcal V,\lambda}\right)^{-1}\widehat T_{\mathcal V}(\lambda)\widehat f$, which were studied in
\cite{AnIsoMo17}. Therefore, in view of \cite{AnIsoMo17} Theorem 7.7, we have for $f \in \big(\mathcal B({\bf T}^2)\big)^2$
\begin{equation}
\begin{split}
& \mathcal U_{\mathcal V} I_2\big(- \widehat\Delta_{\mathcal V} + \cos\sqrt{\lambda \pm i0}\big)^{-1}I_2^{\ast}\mathcal U_{\mathcal V}^{\ast} f  \\
& \simeq \mathcal \sum_{j=1}^2\frac{1}{\lambda_j(x) + \cos\sqrt{\lambda} \mp i\sigma(\lambda)0}\big(P_{\mathcal V,j}(x) f\big)\big|_{ M_{\lambda,j}}.
\end{split}
\label{S5D0lambdaandresolvent}
\end{equation}
We denote the right-hand side as
\begin{equation}
 \big(- \widehat\Delta_{\mathcal V} + \cos\sqrt{\lambda \pm i0}\big)^{-1} f\big|_{ M_{\lambda}}.
\nonumber
\end{equation}
We can then prove the following theorem for $\widehat H^{(0)}_{\mathcal E}$ by using (\ref{S5D0lambdaandresolvent}), and for $\widehat H_{\mathcal E}$ by the formula (\ref{DefineFpmlambdainmathcalE}) and the resolvent equation.
\begin{theorem}
\label{ResolventExpansionPerturbedCase}
For any  $\lambda \in (0,\infty)\setminus{\mathcal T}$ and $\widehat f \in \mathcal B(\mathcal E)$, we have
\begin{equation}
\widehat R_{\mathcal E}(\lambda \pm i0)\widehat f
 \simeq \frac{\sin\sqrt{\lambda}}{\sqrt{\lambda}}\Phi^{(0)}(\lambda)^{\ast}\big(- \widehat\Delta_{\mathcal V} + \cos\sqrt{\lambda \pm i0}\big)^{-1}\big({\mathcal U_{\mathcal V}}^{\ast}\, \widehat{\mathcal F}^{(\pm)}(\lambda)\widehat f\big)\big|_{M_{\lambda}}.
\nonumber
\end{equation}
\end{theorem}

\subsection{Helmholtz equation and S-matrix}
Theorem \ref{ResolventExpansionPerturbedCase} enables us to characterize the solution space to the Helmholtz equation.
\begin{lemma}
\label{LemmaRangeFlambdaast=solutionspace}
Let  $\lambda \in (0,\infty) \setminus \mathcal{T}$ and $\widehat f \in \mathcal B(\mathcal E)$. Then
\begin{equation}
\{\widehat u \in \widehat{\mathcal B}^{\ast}(\mathcal E)\, ; \, 
(\widehat H_{\mathcal E} - \lambda)\widehat u = 0\} 
= \widehat{\mathcal F}^{(-)}(\lambda)^{\ast}{\bf h}_{\lambda}.
\end{equation}
\end{lemma}
\begin{theorem}
\label{HelmhlotzEqWholespace}
For any incoming data $\phi^{in} \in L^2(M_{\lambda})$, there exist a unique solution $\widehat u \in \widehat{\mathcal B}^{\ast}(\mathcal E)$ of the equation 
$$
(\widehat H_{\mathcal E} - \lambda)\widehat u = 0
$$
and an outgoing data $\phi^{out} \in L^2(M_{\lambda})$ satisfying

\begin{equation}
\begin{split}
 \widehat u
 \simeq &  
- \Phi^{(0)}(\lambda)^{\ast}\sum_{j=1}^2\frac{1}{\lambda_j(x) + \cos\sqrt\lambda + i0\sigma(\lambda)}\big(P_{\mathcal V,j}(x)\phi^{in}_j\big)\big|_{M_{\lambda,j}} \\
& +  \Phi^{(0)}(\lambda)^{\ast}\sum_{j=1}^2\frac{1}{\lambda_j(x) + \cos\sqrt\lambda - i0\sigma(\lambda)}\big(P_{\mathcal V,j}(x)\phi^{out}_j\big)\big|_{M_{\lambda,j}}.
\end{split}
\label{Smatrixedge}
\end{equation}
The mapping
\begin{equation}
S(\lambda) : \phi^{in} \to \phi^{out}
\nonumber
\end{equation}
is the S-matrix, which is unitary on ${\bf h}_{\lambda}$.
\end{theorem}

We omit the proof of Lemma \ref{LemmaRangeFlambdaast=solutionspace} and 
Theorem \ref{HelmhlotzEqWholespace}, since they are almost the same as that of  Theorem 7.15 of \cite{AnIsoMo17}.

As is proven in \cite{KoSa15}, using the wave operator
\begin{equation}
\widehat W_{\pm} = {\mathop{\rm s-lim}_{t\to\pm\infty}}\, e^{it\widehat H_{\mathcal E}}e^{-it\widehat H^{(0)}_{\mathcal E}}\widehat P_{ac}(\widehat H^{(0)}_{\mathcal E}),
\nonumber
\end{equation}
where $\widehat P_{ac}(\widehat H^{(0)}_{\mathcal E})$ is the projection onto the absolutely continuous subspace of $\widehat H^{(0)}_{\mathcal E}$, one can define the scattering operator 
$$
\widehat S = (\widehat W^{(+)})^{\ast}\widehat W^{(-)},
$$
which 
is unitary. Define $S$ by
$$
S =\widehat{\mathcal F}^{(0)}\widehat S (\widehat{\mathcal  F}^{(0)})^{\ast}.
$$
The S-matrix $S(\lambda)$ and the scattering amplitude $A(\lambda)$ are defined by
\begin{equation}
S(\lambda) = 1 - 2\pi iA(\lambda),
\nonumber
\end{equation}
\begin{equation}
A(\lambda) = \widehat{\mathcal F}^{(+)}(\lambda)q_{\mathcal E}\widehat{\mathcal F}^{(0)}(\lambda).
\label{S5FormulaScateringamplitude}
\end{equation}
Then $S(\lambda)$ is unitary on ${\bf h}_{\lambda}$, and for $\lambda \in (0,\infty)\setminus{\mathcal T}$
\begin{equation}
(Sf)(\lambda) = S(\lambda)f(\lambda), \quad f \in {\bf H}.
\nonumber
\end{equation}
Since the resolvent has a meromorphic extension into the lower half-plane $\{{\rm Re}\, \lambda > 0, \ {\rm Im}\, \lambda < 0\}$ with possible branch points on $\mathcal T$, the formula (\ref{S5FormulaScateringamplitude}) implies that 
the S-matrix $S(\lambda)$ is also meromorphic in the same domain.


\section{From S-matrix to interior D-N map}
\label{SectionSmatrixtoDNmap}

\subsection{Boundary value problem}
For a subgraph $\Omega = \{\mathcal V_{\Omega}, \mathcal E_{\Omega}\} \subset \{\mathcal V, \mathcal E\}$ and $v \in \mathcal V$, $v \sim \Omega$ means that there exist a vertex $w \in \mathcal V_{\Omega}$ and an  edge ${\bf e} \in \mathcal E$ such that 
$v \sim w$, ${\bf e}(0) = v$ or ${\bf e}(1) = v$. 
For a connected subgraph $\Omega \subset \{\mathcal V, \mathcal E\}$, we define a subset $\partial\Omega = \{\mathcal V_{\partial\Omega}, \mathcal E_{\partial\Omega}\} \subset \{\mathcal V, \mathcal E\}$ by 
\begin{equation}
\mathcal V_{\partial\Omega} = \{v \not\in \mathcal V_{\Omega}\, ; \, 
v \sim {\Omega}\},
\nonumber
\end{equation}
\begin{equation}
\mathcal E_{\partial\Omega} = \{{\bf e} \in \mathcal E\, ; \, {\bf e}(0) \in \mathcal V_{\partial\Omega} \ {\rm or} \ {\bf e}(1) \in \mathcal V_{\partial\Omega}\}.
\nonumber
\end{equation}
We then put $\overline{\Omega} = \Omega \cup \partial\Omega$ and
\begin{equation}
\stackrel{\circ}{\mathcal V_{\overline{\Omega}}} = {\mathcal V}_{\Omega}, \quad \partial \mathcal V_{\overline{\Omega}} = \mathcal V_{\partial\Omega},
\nonumber
\end{equation}
which are called the set of {\it interior vertices}  and the set of {\it boundary vertices} of $\overline{\Omega}$, respectively. We  put
\begin{equation}
\mathcal V_{\overline{\Omega}} = \stackrel{\circ}{\mathcal V_{\overline{\Omega}}}\cup\, \partial\mathcal V_{\overline{\Omega}}.
\nonumber
\end{equation}
As for the edges, we simply put
\begin{equation}
{\mathcal E}_{\overline{\Omega}} = {\mathcal E}_{\Omega}\cup \mathcal E_{\partial\Omega}.
\nonumber
\end{equation}
We then define the {\it edge Dirichlet Laplacian} $\widehat\Delta_{\mathcal E_{\overline{\Omega}}}$ by
\begin{equation}
\widehat\Delta_{\mathcal E_{\overline{\Omega}}} u_{\bf e}(z) = \frac{d^2}{dz^2}u_{\bf e}(z), \quad 
{\bf e} \in \mathcal E_{\overline{\Omega}}
\nonumber
\end{equation}
whose domain $D(\widehat\Delta_{\mathcal E_{\overline{\Omega}}})$ is the set of all $u = \{u_{\bf e}\}_{{\bf e} \in \mathcal E_{\overline{\Omega}}} \in H^2(\mathcal E_{\overline{\Omega}})$ satisfying $u(v) = 0$ at any boundary vertex $v \in \partial \mathcal V_{\overline{\Omega}}$ and the Kirchhoff condition at any interior vertex $v \in \stackrel{\circ}{\mathcal V}_{\overline{\Omega}}$. By the standard argument, $\widehat\Delta_{\mathcal E_{\overline{\Omega}}}$ is self-adjoint.

The {\it vertex Dirichlet Laplacian on $\mathcal V_{\overline{\Omega}}$}  is defined in the same way as in (\ref{S3DefineVertexLaplacian}) :
\begin{equation}
\big(\widehat\Delta_{\mathcal V_{\overline{\Omega}}, \lambda}\widehat u\big)(v) = \frac{1}{{\rm deg}_{\mathcal V_{\overline{\Omega}}}(v)}
\sum_{w\sim v, w\in \mathcal V_{\overline{\Omega}}}\frac{1}{\psi_{wv}(1,\lambda)}\widehat u(w), \quad 
v \in \mathcal V_{\overline{\Omega}}.
\nonumber
\end{equation}
Recall that for a domain $\mathcal W \subset \mathcal V$, we define
\begin{equation}
\deg_{\mathcal W}(v) = 
\left\{
\begin{split}
& \sharp\, \{w \in \mathcal W \, ; \, w\sim v\}, \quad  v \in \stackrel{\circ}{\mathcal W}, \\
& \sharp\, \{w \in \stackrel{\circ}{\mathcal W} \, ; \, w\sim v\}, \quad  v \in \partial{\mathcal W}.
\end{split}
\right.
\nonumber
\end{equation}
(See (2.6) of \cite{AnIsoMo17(1)}). 
We impose the Dirichlet boundary condition for the domain $D(\widehat\Delta_{\mathcal V_{\overline{\Omega}}, \lambda})$ : 
\begin{equation}
\widehat u \in D(\widehat\Delta_{\mathcal V_{\overline{\Omega}}, \lambda}) \Longleftrightarrow \widehat u \in \ell^2(\mathcal V_
{\overline{\Omega}})\cap\{\widehat u \, ; \, \widehat u(v) = 0, \ v \in \partial \mathcal V_{\overline{\Omega}}\}.
\nonumber
\end{equation}
As in \S 3, we first define the vertex Dirichlet Laplacian for the  case without potential and then add the pontntial $\widehat Q_{\mathcal V,\lambda}$ as a perturbation. By modifying the inner product, $- \widehat\Delta_{\mathcal V_{\overline{\Omega}},\lambda} + \widehat Q_{\mathcal V,\lambda}$ is self-adjoint.
 The normal derivative at the boundary associated with $\widehat\Delta_{\mathcal V_{\overline{\Omega}},\lambda}$ is defined by
\begin{equation}
\big(\partial^{\nu}_{\widehat\Delta_{\mathcal V_{\overline{\Omega}},\lambda}}\widehat u\big)(v) = - \frac{1}{{\rm deg}_{\mathcal V_{\overline{\Omega}}}(v)}\sum_{w\sim v, w \in \stackrel{\circ}{\mathcal V_{\overline{\Omega}}}}\frac{1}{\psi_{wv}(1,\lambda)}
\widehat u(w).
\label{DefinenormalderivativeVDlambda}
\end{equation}
(c.f. (2.7) of \cite{AnIsoMo17(1)}).
Note that in the right-hand side, $w$ is taken only from $\stackrel{\circ}{\mathcal V_{\overline{\Omega}}}$.

Let us give an example of interior and exterior domains as well as their boundaries for the case of hexagonal lattice. 
We identify ${\bf R}^2$ with ${\bf C}$ and put $\omega = e^{2\pi i/6} = (1 + \sqrt{3}i)/2$. 
Let $\mathcal D$ be the hexagon with center at the origin and vertices $\omega^n, 1 \leq n \leq 6$. Recalling that the basis of the hexagonal lattice are $2 - \omega$ and $1 + \omega$, we put
\begin{equation}
\mathcal D_{k\ell} = \mathcal D + k(2-\omega) + \ell(1 + \omega),
\nonumber
\end{equation}
which denotes the translation of $\mathcal D$ by $k(2-\omega)$ and $\ell(1 + \omega)$.
For an integer $L \geq 1$, let 
\begin{equation}
\mathcal D_L = {\mathop\cup_{|k|\leq L, |\ell| \leq L}}\mathcal D_{k\ell}.
\nonumber
\end{equation}
 As is illustrated in Figure \ref{BoundaryHexagonal}, we take an interior domain $\Omega_{int}$ in such a way that
\begin{equation}
\stackrel{\circ}{\mathcal V_{\Omega_{int}}} = \mathcal V\cap \mathcal D_L, \quad 
\stackrel{\circ}{\mathcal E_{\Omega_{int}}} = \mathcal E\cap \mathcal D_L.
\nonumber
\end{equation}
In Figure \ref{BoundaryHexagonal}, $\partial \mathcal V_{\Omega_{int}}$ is denoted by white dots. 

\begin{figure}[h]
\centering
\includegraphics[width=5cm, bb=0 0 606 564]{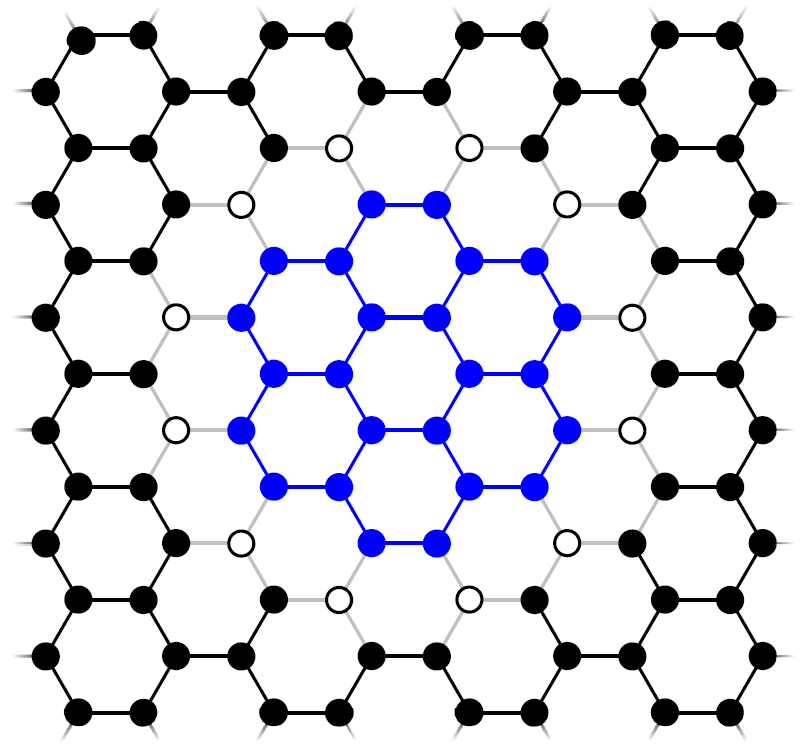}
\caption{Boundary of a domain in the hexagonal lattice}
\label{BoundaryHexagonal}
\end{figure}

The exterior domain $\Omega_{ext}$ is defined similarly. We then put
$$
\mathcal V_{int} = \mathcal V_{\overline{\Omega_{int}}}, \quad \mathcal E_{int} = \mathcal E_{\overline{\Omega_{int}}}, 
$$
$$
\mathcal V_{ext} = \mathcal V_{\overline{\Omega_{ext}}}, \quad \mathcal E_{ext} = \mathcal E_{\overline{\Omega_{ext}}}, 
$$
for the sake of simplicity. 
Note that
\begin{equation}
\mathcal V = \mathcal V_{int} \cup \mathcal V_{ext}, \quad 
\partial\mathcal V_{int} = \partial \mathcal V_{ext},
\nonumber
\end{equation}
\begin{equation}
\mathcal E = \mathcal E_{int}\cup\mathcal E_{ext}, \quad 
\mathcal E_{int}\cap\mathcal E_{ext} = \emptyset.
\nonumber
\end{equation}
We define the edge Dirichlet Laplacians on $\mathcal E_{int}$, $\mathcal E_{ext}$, which are denoted by $\widehat\Delta_{int,\mathcal E}$, $\widehat\Delta_{ext,\mathcal E}$:
$$
\widehat \Delta_{int,\mathcal E} = \widehat \Delta_{\mathcal E_{int}}, \quad \widehat \Delta_{ext,\mathcal E} = \widehat\Delta_{\mathcal E_{ext}}.
$$ 
Let us note that 
$$
\sigma_e\big(- \widehat{\Delta}_{\mathcal E}\big) = \sigma_e\big(- \widehat{\Delta}_{ext,\mathcal E}\big).
$$

We assume that the support of the potential lies strictly inside of  $\mathcal E_{int}$. Namely introducing a set:
\begin{equation}
\widetilde{\mathcal E_{int}} = \{{\bf e} \in \mathcal E_{int}\, ; \, {\bf e}(0) \not\in \partial\mathcal V_{int}, \ {\bf e}(1) \not\in \partial\mathcal V_{int}\},
\nonumber
\end{equation}
we assume
\begin{equation}
{\rm supp}\, q_{\mathcal E} \subset \widetilde{\mathcal E_{int}}.\label{SupportCond}
\end{equation}

The formal formulas (\ref{Resolventformalformula1}), (\ref{Resolventformula2})  are also valid for  boundary value problems of  edge Laplacians. For the case of the exterior problem, the resolvent of $- \widehat{\Delta}_{ext,\mathcal E}$ is written by (\ref{Resolventformula2}) with $\widehat H^{(0)}_{\mathcal E}$ replaced by $- \widehat{\Delta}_{ext,\mathcal E}$.
In our previous work \cite{AnIsoMo17(1)}, we studied the spectral properties of the vertex Laplacian in the exterior domain by reducing them to the whole space problem. Therefore, all the results for the edge Laplacian in the previous section also hold in the exterior domain. In particular, we have 
\begin{itemize}
\item Rellich type theorem (Theorem \ref{RelichTypeTheorem}),
\item Limiting absorption principle (Theorem \ref{LAPwholespace}),
\item Spectral representation (Theorem \ref{EigenfunctionExpansionWholespace}),
\item Resolvent expansion (Theorem \ref{ResolventExpansionPerturbedCase}),
\item Exapansion of solutions to the Helmholtz equation  (Theorem \ref{HelmhlotzEqWholespace}), 
\item S-matrix (Theorem \ref{HelmhlotzEqWholespace})
\end{itemize}
in the exterior domain $\mathcal E_{ext}$. In fact, Theorem \ref{RelichTypeTheorem} holds without any change. Using the formula (\ref{widehatue(z,lambda)=PhiC+ref}) and the limiting absorption principle for $\widehat\Delta_{ext}$ proven in Theorem 7.7 in \cite{AnIsoMo17}, one can extend Theorem \ref{LAPwholespace} for the exterior domain. 
The radiation condition is also extended to the exterior domain. Then, the remaining theorems 
(Theorems \ref{EigenfunctionExpansionWholespace}, \ref{HelmhlotzEqWholespace}) are proven by the same argument.
\subsection{Exterior and interior D-N maps}

We consider the edge model for the exterior problem. 
Let
$\widehat u^{(\pm)} = \{\widehat u^{(\pm)}_{\bf e}\}_{{\bf e}\in \mathcal E_{ext}}$ be the solution to the equation
\begin{equation}
\left\{
\begin{split}
& ( - \widehat{\Delta}_{ext,\mathcal E} - \lambda)\widehat u = 0, \quad 
{\rm in} \quad \stackrel{\circ}{{\mathcal E}_{ext}}, \\
& \widehat u = \widehat f, \quad {\rm on} \quad \partial \mathcal E_{ext},
\end{split}
\right.
\label{EdgeEquationExterior}
\end{equation}
 satisfying the radiation
condition (outgoing for $\widehat u^{(+)}$ and incoming for $\widehat u^{(-)}$).
Then, the extrior D-N map $\Lambda^{(\pm)}_{ext,\mathcal E}(\lambda)$ is defined by
\begin{equation}
\Lambda^{(\pm)}_{ext,\mathcal E}(\lambda)\widehat f(v) = - \frac{d}{dz}\widehat u^{(\pm)}_{\bf e}(v), \quad v \in \partial \mathcal V_{ext}, 
\label{DefineLambdaextE}
\end{equation}
where ${\bf e}$ is the edge having $v$ as its end point.
Here, to compute $\frac{d}{dz}\widehat u_{\bf e}^{(\pm)}(v)$, we negelect the original orientation of ${\bf e}$. Namely, we parametrize ${\bf e}$ by $z \in [0,1]$ so that $v \in \partial\mathcal V$ corresponds to $z= 0$, and define $\frac{d}{dz}\widehat u^{(\pm)}_{\bf e}(v) = \frac{d}{dz}\widehat u^{(\pm)}_{\bf e}(z)\big|_{z=0}$.

 For the case of the interior problem, the Dirichlet boundary value problem for the edge Laplacian 
\begin{equation}
\left\{
\begin{split}
& ( - \widehat{\Delta}_{int,\mathcal E} + q_{{\bf e}} - \lambda)\widehat u = 0, \quad 
{\rm in} \quad \stackrel{\circ}{{\mathcal E}_{int}}, \\
& \widehat u = \widehat f, \quad {\rm on} \quad \partial \mathcal V_{int}
\end{split}
\right.
\label{EdgeEquationInterior}
\end{equation}
is formulated as above.
Note that the spectrum of $- \widehat\Delta_{int,\mathcal E} + q_{\mathcal E}$ is discrete. In the following, we assume that 
\begin{equation}
\lambda \not\in \sigma(- \widehat\Delta_{int,\mathcal E} + q_{\mathcal E}).
\label{lambdanotinsigamDeltaintE}
\end{equation}
The D-N map $\Lambda_{int,\mathcal E}(\lambda)$ is defined by
\begin{equation}
\Lambda_{int,\mathcal E}(\lambda)\widehat f(v) = \frac{d}{dz}\widehat u_{\bf e}(v), \quad v \in \partial \mathcal V_{int}, 
\label{LambdaintE}
\end{equation}
where ${\bf e}$ is the edge having $v$ as its end point and $\widehat u = \{\widehat u^{(\pm)}_{\bf e}\}_{{\bf e}\in \mathcal E_{int}}$ 
is the solution to the equation (\ref{EdgeEquationInterior}).
The same remark as above is applied to $\frac{d}{dz}\widehat u_{\bf e}(v), v \in \partial\mathcal V_{int}$.

The D-N maps are also defined for vertex operators. 
Let us slightly change the  notation. 
For a subset $\mathcal V_D \subset \mathcal V$ and $v \in \mathcal V_D$, let
\begin{equation}
(\widehat\Delta^{(0)}_{\mathcal V}\widehat u)(v) = \frac{1}{3}
\sum_{w \sim v}\widehat u(w),
\nonumber
\end{equation}
\begin{equation}
(\widehat\Delta^{(0)}_{\mathcal V_D}\widehat u)(v) = \frac{1}{{\rm deg}_{\mathcal V_D}(v)}
\sum_{w \sim v, w \in \mathcal V_D}\widehat u(w).
\nonumber
\end{equation}
By this definition, we have (see (\ref{DefineDeltaVlambda(0)}))
\begin{equation}
\widehat\Delta^{(0)}_{\mathcal V,\lambda} = \frac{\sqrt{\lambda}}{\sin\sqrt{\lambda}}\widehat\Delta^{(0)}_{\mathcal V}.
\label{Delta0Vlambda=fracsinlambdalambdaDelta0}
\end{equation}
For the exterior and interior domains $\Omega_{ext}$ and $\Omega_{int}$ defined in the previous section, $\widehat\Delta_{\mathcal V_D}^{(0)}$ is denoted by $\widehat\Delta_{ext,\mathcal V}$ and $\widehat\Delta_{int,\mathcal V}$, respectively:
\begin{equation}
\widehat\Delta_{ext,\mathcal V} = \widehat\Delta^{(0)}_{\mathcal V_{ext}}, \quad 
\widehat\Delta_{int,\mathcal V} = \widehat\Delta^{(0)}_{\mathcal V_{int}}.
\nonumber
\end{equation}

Now, consider the exterior boundary value problem
\begin{equation}
\left\{
\begin{split}
& \big(- \widehat\Delta_{\mathcal V,\lambda} + \widehat Q_{\mathcal V,\lambda}\big)\widehat u = 0, \quad {\rm in} \quad \stackrel{\circ}{\mathcal V_{ext}}, \\
& \widehat u = \widehat f, \quad {\rm on} \quad \partial\mathcal V_{ext}.
\end{split}
\right.
\label{ExteriorBVPVertexmodel}
\end{equation}
Note that by (\ref{Delta0Vlambda=fracsinlambdalambdaDelta0}) and (\ref{DefinewidehatQ0}) this is equivalent to 
\begin{equation}
\left\{
\begin{split}
& (- \widehat\Delta^{(0)}_{\mathcal V} + \cos\sqrt{\lambda})\widehat u = 0, \quad {\rm in} \quad \stackrel{\circ}{\mathcal V_{ext}}, \\
& \widehat u = \widehat f, \quad {\rm on} \quad \partial\mathcal V_{ext}.
\end{split}
\right.
\nonumber
\end{equation}
Let $\widehat u^{(\pm)}_{ext,\mathcal V}$ be the solution of this equation satisfying the radiation condition.
Then, taking account of (\ref{DefinenormalderivativeVDlambda}) and  (\ref{Delta0Vlambda=fracsinlambdalambdaDelta0}), 
we define the exterior D-N map by
\begin{equation}
\begin{split}
\widehat \Lambda^{(\pm)}_{ext,\mathcal V}(\lambda)\widehat f &= 
- \frac{\sin\sqrt{\lambda}}{\sqrt{\lambda}}\partial^{\nu}_{\widehat\Delta_{\mathcal V_{ext},\lambda}}\widehat u^{(\pm)}_{ext,\mathcal V} = \partial^{\nu}_{\widehat\Delta_{ext,\mathcal V}}\widehat u^{(\pm)}_{ext,\mathcal V} \\
& = \frac{1}{{\rm deg}_{\mathcal V_{ext}}(v)}\sum_{w\sim v, w\in \stackrel{\circ}{\mathcal V_{ext}}}\widehat u^{(\pm)}_{ext,\mathcal V}(w).
\end{split}
\label{LambdaextV}
\end{equation}

We also consider the interior boundary value problem
\begin{equation}
\left\{
\begin{split}
& \big(- \widehat\Delta_{\mathcal V,\lambda} + \widehat Q_{\mathcal V,\lambda}\big)\widehat u = 0, \quad {\rm in} \quad \stackrel{\circ}{\mathcal V_{int}}, \\
& \widehat u = \widehat f, \quad {\rm on} \quad \partial\mathcal V_{int}.
\end{split}
\right.
\label{IntBVPVertex}
\end{equation}
Taking account of (\ref{SupportCond}), we define the interior D-N map by
\begin{equation}
\begin{split}
\widehat \Lambda_{int,\mathcal V}(\lambda)\widehat f(v) &= 
 \frac{\sin\sqrt{\lambda}}{\sqrt{\lambda}}\partial^{\nu}_{\widehat\Delta_{\mathcal V_{int},\lambda}}\widehat u_{int,\mathcal V} = - \partial^{\nu}_{\widehat\Delta_{int,\mathcal V}}\widehat u_{int,\mathcal V} \\
& = - \frac{1}{{\rm deg}_{\mathcal V_{int}}(v)}\sum_{w\sim v, w\in \stackrel{\circ}{\mathcal V_{int}}}\widehat u_{int,\mathcal V}(w).
\end{split}
\label{DefineLambdaintV}
\end{equation}

Note that by virtue of Lemma \ref{Lemma3.2}, if $\widehat u$ satisfies the edge Schr{\"o}dinger equation $(\widehat H_{\mathcal E} - \lambda)\widehat u = 0$ and the Kirchhoff condition, $\widehat u\big|_{\mathcal V}$ satisfies the vertex Schr{\"o}dinger equation $\big(- \widehat \Delta_{\mathcal V,\lambda} + \widehat Q_{\mathcal V,\lambda}\big)\widehat u\big|_{\mathcal V} = 0$. 
Therefore, if the exterior boundary value problem (\ref{EdgeEquationExterior}) for the edge model is solvable, so is the exterior boundary value problem  
(\ref{ExteriorBVPVertexmodel}) for the vertex model. The same remark applies to the interior boundary value problem.

If $\varphi(z)$ satisfies $- \varphi''(z) - \lambda\varphi(z) = 0$ in $(0,1)$, we have
$$
\varphi(1) = \varphi(0)\cos\sqrt{\lambda} +  \varphi'(0)\frac{\sin\sqrt{\lambda}}{\sqrt{\lambda}},
$$
Since the D-N map for the vertex model is computed by $\widehat u\big|_{\mathcal V}$, where $\widehat u$ is the solution to the edge Schr{\"o}dinger equation, this implies, by (\ref{DefineLambdaextE}), (\ref{LambdaintE}), (\ref{LambdaextV}) and (\ref{DefineLambdaintV}), the following formulas between the D-N maps of edge-Laplacian and vertex Laplacian.
\begin{lemma} 
\label{LemmaDNedge=DNvertex}
The following equalities hold: 
\begin{equation}
\widehat\Lambda^{(\pm)}_{ext,\mathcal V}(\lambda) = \cos\sqrt{\lambda} - \frac{\sin\sqrt{\lambda}}{ \sqrt{\lambda}}\, \Lambda^{(\pm)}_{ext,\mathcal E}(\lambda), \quad 
\lambda \in (0,\infty)\setminus \mathcal T,
\nonumber
\end{equation}
\begin{equation}
\widehat\Lambda_{int,\mathcal V}(\lambda) = - \cos\sqrt{\lambda} -  \frac{\sin\sqrt{\lambda}}{\sqrt{\lambda}}\, \Lambda_{int,\mathcal E}(\lambda), \quad  \lambda \in {\bf C}\setminus \sigma(- \widehat{\Delta}_{int,\mathcal E} + q_{\mathcal E}).
\nonumber
\end{equation}
\end{lemma}

Therefore, the D-N map for the edge model and the D-N map for the vertex model determine each other.


\subsection{Relations between S-matices and D-N maps}
 We show that the S-matrices for the vertex Laplacian and the edge Laplacian coincide.

In \cite{AnIsoMo17}, Theorem 7.15, we have proven the following theorem, which is the counter part of Theorem \ref{HelmhlotzEqWholespace} for the discrete Laplacian $- \widehat\Delta_{\mathcal V}$ at the energy $- \cos\sqrt{\lambda}$: 
For any incoming data $\phi^{in} \in L^2(M_{\lambda})$, there exist a unique solution $\widetilde u_{\mathcal V} \in \widehat{\mathcal B}^{\ast}(\mathcal V)$ of the equation 
$$
(- \widehat\Delta_{\mathcal V} + \cos \sqrt{\lambda})\widetilde u_{\mathcal V} = 0
$$
and an outgoing data $\widetilde \phi^{out} \in L^2(M_{\lambda})$ satisfying
\begin{equation}
\begin{split}
I_2 \widetilde u_{\mathcal V}
 \simeq &  
- \sum_{j=1}^2\frac{1}{\lambda_j(x) + \cos\sqrt\lambda + i0\sigma(\lambda)}\big(P_{\mathcal V,j}(x)\phi^{in}_j\big)\big|_{M_{\lambda,j}} \\
& +  \sum_{j=1}^2\frac{1}{\lambda_j(x) + \cos\sqrt\lambda - i0\sigma(\lambda)}\big(P_{\mathcal V,j}(x)\widetilde \phi^{out}_j\big)\big|_{M_{\lambda,j}}.
\end{split}
\label{Smatrixvertex}
\end{equation}
in the sense that the difference of both sides is in $\mathcal B_0^{\ast}({\bf T}^2;{\bf C}^2)$.
The mapping
\begin{equation}
\widetilde S(\lambda) : \phi^{in} \to \widetilde\phi^{out}
\nonumber
\end{equation}
is the S-matrix of $- \widehat\Delta_{\mathcal V}$ at the energy $- \cos\sqrt{\lambda}$, which is unitary on ${\bf h}_{\lambda}$.

By virtue of Lemma \ref{S2B0astEandB0astVLemma}, we see that 
$\widetilde u_{\mathcal E} := \Phi^{(0)}(\lambda)^{\ast}\widetilde u_{\mathcal V}$ has the properties in Theorem \ref{HelmhlotzEqWholespace}, hence by the uniqueness $u = \widetilde u_{\mathcal E}$. Therefore, $\widetilde \phi^{(out)} = \phi^{(out)}$, which implies $S(\lambda) = \widetilde S(\lambda)$. 
We have thus proven the followin theorem.

\begin{theorem}
\label{VertexSmatrix=EdgeSmatrix}
The S-matrix for the vertex Schr{\"o}dinger operator  at the energy $- \cos\sqrt{\lambda}$ coincides with that of  the edge Schr{\"o}dinger operator  at the energy $\lambda$. 
\end{theorem}

In \cite{AnIsoMo17(1)}, we have proven that for the vertex Laplacian the scattering matrix and the interior D-N map determine each other.   By virtue of Theorems \ref{LemmaDNedge=DNvertex} and \ref{VertexSmatrix=EdgeSmatrix}, we have the following theorem. 
\begin{theorem}
\label{Aext-A=IBSigma-1ILemma}
For the edge Laplacian on the hexagonal lattice, the S-matrix and the D-N map in the interior domain determine each other. 
\end{theorem}



\section{Inverse scattering}
\label{SectionInversescattering}
\subsection{Hexagonal parallelogram}
We are now in a position to consider the inverse scattering problem. 
Note here that although
 the choice of  fundamental domain of the lattice $\mathcal L$ is not unique, different choices give rise to unitarily equivalent Hamiltonians. 
 In this section, we take ${\bf v}_1, {\bf v}_2$ and  $p^{(1)}, p^{(2)}$ as in (\ref{v1andv2inthepreviouspaper}) and  (\ref{p1andp2inthepreviouspaper}) to make use of our previous results in \cite{AnIsoMo17},  \cite{AnIsoMo17(1)}. 
We identify ${\bf R}^2$ with ${\bf C}$, and put
$$
\omega = e^{\pi i/3}.
$$
 For $n = n_1 + i n_2 \in {\bf Z}[i]= {\bf Z} + i{\bf Z}$, 
let
$$
\mathcal L_0 = \left\{{\bf v}(n)\, ; \, n \in {\bf Z}[i]\right\}, \quad
{\bf v}(n) = n_1{\bf v}_1 + n_2{\bf v}_2,
$$
\begin{equation}
{\bf v}_1 = 1 + \omega, \quad {\bf v}_2 = \sqrt3 i,
\label{v1andv2inthepreviouspaper}
\end{equation}
\begin{equation}
p_1 = \omega^{-1} = \omega^5, \quad p_2 = 1,
\label{p1andp2inthepreviouspaper}
\end{equation}
and define the vertex set $\mathcal V_0$ by
$$
\mathcal V_0 = \mathcal V_{01} \cup \mathcal V_{02}, \quad \mathcal V_{0i} = p_i + \mathcal L_0.
$$

By virtue of Theorem \ref{Aext-A=IBSigma-1ILemma}, given an S-matrix and a bounded domain $\mathcal E_{int}$, we can compute the D-N map 
associated with $\mathcal E_{int}$.  The problem is now reduced to the reconstruction of the potentials on the edges from the knowledge of the D-N map for the vertex Schr{\"o}dinger operator defined on $\mathcal V_{int}$,  the set of the vertices in $\mathcal E_{int}$.

As $\mathcal V_{int}$, we use the following domain which is different from the one in Figure \ref{BoundaryHexagonal}.
Let $\mathcal D_0$ be the Wigner-Seitz cell of $\mathcal V_0$.
 It is a hexagon
 having 6 vertices 
 $\omega^k,\ 0 \leq k \leq 5$,
 with center at the origin. Take $D_N = \{n \in {\bf Z}[i]\, ; \, 0 \leq n_1 \leq N, \  0 \leq n_2 \leq N\}$, where $N$ is chosen large enough,  and put
 $$
 \mathcal D_N = {\mathop\cup_{n \in D_N}}\Big( \mathcal D_0 + {\bf v}(n)\Big).
 $$
 This is a parallelogram in the hexagonal lattice (see Figure \ref{S6HexaParallel}).
\begin{figure}[h]
\includegraphics[width=7cm, bb=0 0 664 676]{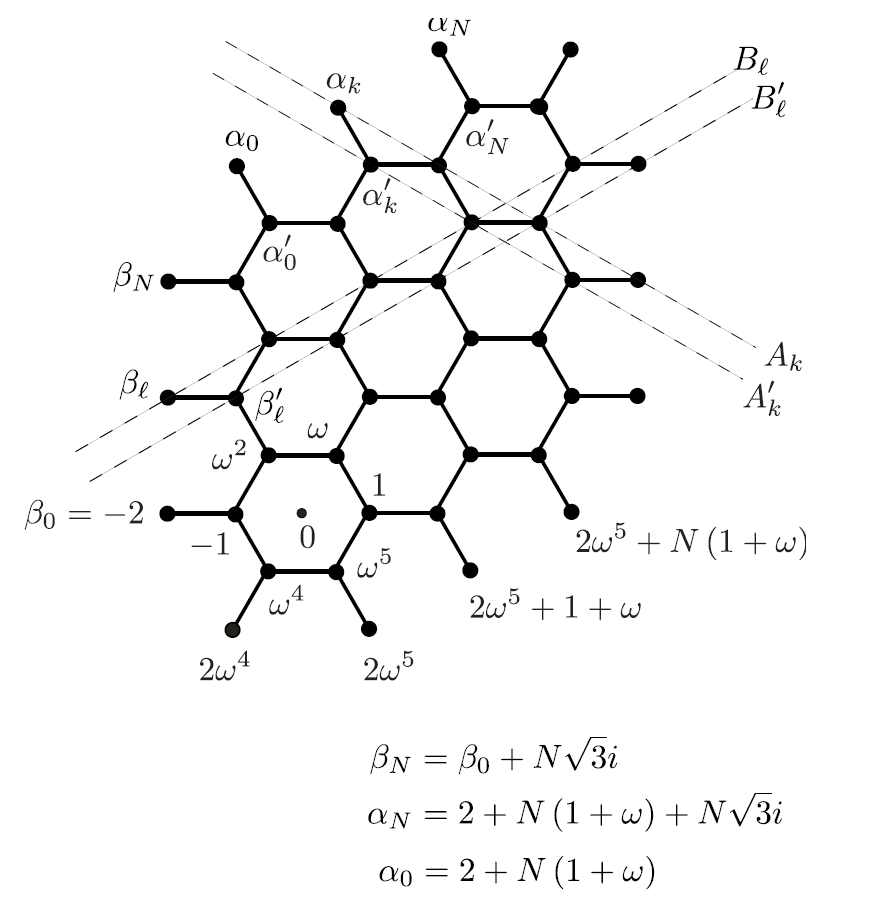}
\caption{Hexagonal parallelogram ($N = 2$)}
\label{S6HexaParallel}
\end{figure}
The interior angle of each vertex on the periphery of $\mathcal D_N$ is either $2\pi/3$ or $4\pi/3$. Let $\mathcal A$ be the set of the former, and for  each $z \in \mathcal A$, we assign a new edge $ e_{z,\zeta}$, and a new vertex $\zeta = t(e_{z,\zeta})$ on its terminal point, hence $\zeta$ is in the outside of $\mathcal D_N$. 
Let 
$$
\Omega = \{v \in \mathcal V_0\, ; \, v \in \mathcal D_N\}
$$
 be the set of vertices in the inside of the resulting graph.
The boundary $\partial\Omega = \{t(e_{z,\zeta})\, ; \, z \in \mathcal A\}$ is  divided into 4 parts,   called  top, 
 bottom, right,  left sides, which are denoted by 
$(\partial\Omega)_T, (\partial\Omega)_B, (\partial\Omega)_R, (\partial\Omega)_L$, i.e. 
\begin{gather*}
\begin{split}
 (\partial\Omega)_T =& \{ \alpha_0 , \cdots , \alpha_N \}, \\
 (\partial\Omega)_B = & \{2\omega^5 + k(1 + \omega)\,  ; \, 0 \leq k \leq N\}, \\
 (\partial\Omega)_R = & \{ 2+ N(1+\omega ) + k\sqrt{3} i \, ; \, 1 \leq k \leq N \} \cup \{ 2+N(1+\omega ) +N\sqrt{3} i + 2 \omega ^2  \} , \\
 (\partial\Omega)_L =& \{2\omega^4\}\cup\{\beta_0,\cdots,\beta_N\} , 
\end{split}
\end{gather*}
where $ \alpha_k = \beta _N + 2\omega + k (1+\omega ) $ and $ \beta _k = -2 + k\sqrt{3} i $ for $ 0\leq k \leq N$.

\subsection{Special solutions to the vertex Schr{\"o}dinger equation}
Taking $N$ large enough so that $\mathcal D_N$ contains all the supports of the potentials $q_{\bf e}(z)$ in its interior, 
we consider the following Dirichlet problem for the vertex Schr{\"o}dinger equation
\begin{equation}
\left\{
\begin{split}
& ( - \widehat{\Delta}_{\mathcal V,\lambda} + \widehat Q_{\mathcal V,\lambda})\widehat u = 0, \quad 
{\rm in} \quad \stackrel{\circ}{{\Omega}}, \\
& \widehat u = \widehat f, \quad {\rm on} \quad \partial \Omega.
\end{split}
\right.
\end{equation}
Let  ${\bf\Lambda_{\widehat{\bf Q}}}$ be the associated D-N map. 
The key to the inverse procedure is the following partial data problem.


\begin{lemma}\label{S6partialDNdata}
(1) Given a partial Dirichlet data $\widehat f$ on $\partial\Omega\setminus(\partial \Omega)_R$, and a partial Neumann data $\widehat g$ on $(\partial\Omega)_L$, there is a unique solution $\widehat u$ on $\stackrel{\circ}\Omega \cup (\partial\Omega)_R$ to the equation
\begin{equation}
\left\{
\begin{split}
& (- \widehat\Delta_{\mathcal V,\lambda} + \widehat Q_{\mathcal V,\lambda})\widehat u = 0, \quad {\it in} \quad \stackrel{\circ}\Omega,\\
& \widehat u =\widehat f, \quad {\it on} \quad \partial\Omega\setminus(\partial\Omega)_R, \\
& \partial_{\nu}^{\mathcal D_N}\widehat u = \widehat g, \quad {\it on} \quad (\partial\Omega)_L.
\end{split}
\right.
\label{Lemma61Equation}
\end{equation}
\noindent
(2) Given the D-N map ${\bf\Lambda_{\widehat{\bf Q}}}$, a partial Dirichlet data $\widehat f_2$ on $\partial\Omega\setminus(\partial\Omega)_R$ and a partial Neumann data $\widehat g$ on $(\partial\Omega)_{L}$, there exists a unique $\widehat f$ on $\partial\Omega$ such that $\widehat f = \widehat f_2$ on $\partial\Omega\setminus(\partial\Omega)_R$ and ${\bf\Lambda_{\widehat{\bf Q}}}\widehat f = \widehat g$ on $(\partial\Omega)_{L}$.
\end{lemma}

For the proof, see \cite{AnIsoMo17(1)}, Lemma 6.1.

Now, for $0 \leq k \leq N$, let us consider a {\it diagonal} line $A_k$ (see  Figure \ref{LineAk})  :
\begin{equation}
A_k = \{x_1 + ix_2\, ; \, x_1 + \sqrt3x_2 = a_k\},
\label{S6Dinagonalline}
\end{equation}
where $a_k$ is chosen so that $A_k$ passes through
\begin{equation}
\alpha_k =  \alpha_0 + k (1+ \omega )  \in 
(\partial\Omega)_T.
\end{equation} 
The vertices on $A_k\cap\Omega$ are written as
\begin{equation}
\alpha_{k,\ell} = \alpha_k + \ell (1 + \omega^5 ), \quad \ell = 0, 1, 2, \cdots.
\end{equation}

\begin{figure}[h]
\centering
\includegraphics[width=6cm, bb=0 0 470 482]{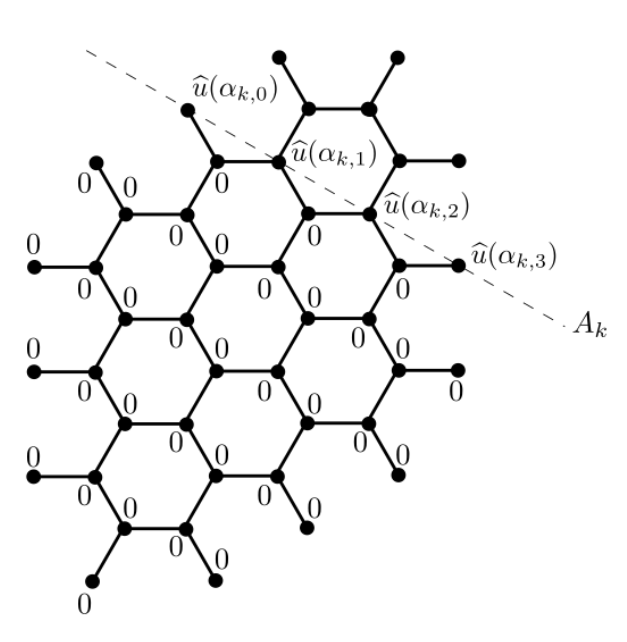}
\caption{Line $A_k$}
\label{LineAk}
\end{figure}

\begin{lemma}
\label{Lemmaspecialboundarydata}
Let $A_k \cap\partial\Omega = \{\alpha_{k,0},\alpha_{k,m}\}$. Then, there exists a unique solution $\widehat u$ to the equation
\begin{equation}
\big(- \widehat\Delta_{\mathcal V,\lambda} + \widehat Q_{\mathcal V,\lambda}\big)\widehat u = 0 \quad {\it in} \quad \stackrel{\circ}\Omega,
\label{S7BVP}
\end{equation}
with partial Dirichlet data $\widehat f$ such that
\begin{equation}
\left\{
\begin{split}
& \widehat f(\alpha_{k,0}) = 1, \\
& \widehat f(z) = 0 \quad {\it for} \quad 
z \in \partial\Omega\setminus\big((\partial\Omega)_R
\cup\alpha_{k,0}\cup\alpha_{k,m}\big)
\end{split}
\right.
\end{equation}
and partial Neumann data $\widehat g = 0$ on $(\partial\Omega)_L$. It satisfies
\begin{equation}
\widehat u(x_1 + ix_2) = 0 \quad {\it if} \quad 
x_1 + \sqrt3 x_2 < a_k.
\label{S7Lemmu=0belowAk}
\end{equation}
\end{lemma}

An important feature is that $\widehat u$ vanishes below the line $A_k$. By using this property, we reconstructed the vertex potentials and defectes of the hexagonal lattice in \cite{AnIsoMo17(1)}. 
We make use of the same idea.

Let $\widehat u$ be a solution of the equation
\begin{equation}
 (- \widehat\Delta_{\mathcal V,\lambda} + \widehat Q_{\mathcal V,\lambda})\widehat u = 0, \quad {\rm in} \quad \stackrel{\circ}\Omega,
\label{S7Equationtobeevaluated}
\end{equation}
which vanishes in the region $x_1 + \sqrt{3}x_2 < a_k$. Let $a, b, b', c \in \mathcal V$ and ${\bf e}, {\bf e}' \in \mathcal E$ be as in  Figure
\ref{S7hex_edge}. 

\begin{figure}[h]
\includegraphics[width=5cm, bb=0 0 579 425]{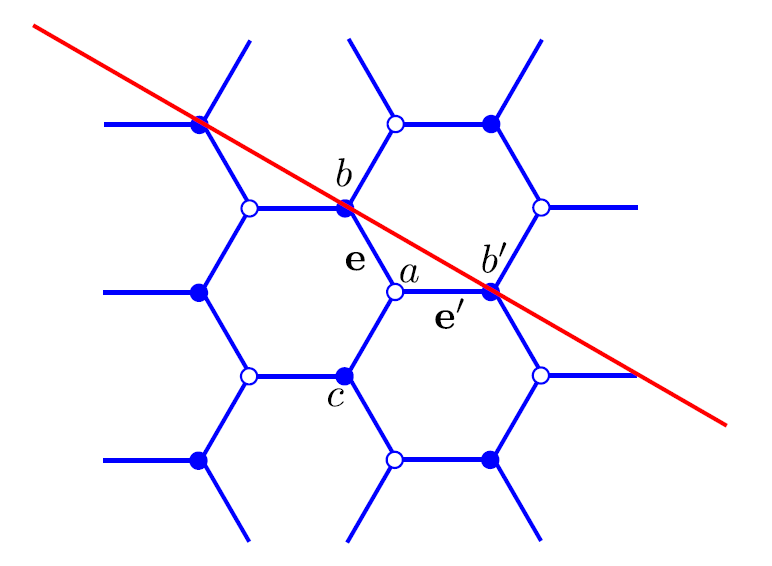}
\caption{$\widehat u(b)$ and $\widehat u(b')$}
\label{S7hex_edge}
\end{figure}

Then, evaluating the equation (\ref{S7Equationtobeevaluated}) at $v = a$ and using (\ref{S3DefineVertexLaplacian}), (\ref{S3DefinewidehatQVlambda}),
we obtain
\begin{equation}
\frac{1}{\psi_{ba}(1,\lambda)}\widehat u(b) + \frac{1}{\psi_{b'a}(1,\lambda)}\widehat u(b') = 0.
\label{S7equationbypsiwv}
\end{equation}
Here, for any edge ${\bf e} \in \mathcal E$, we associate an edge $[{\bf e}]$ without orientation and a function $\phi_{[{\bf e}]}(z,\lambda)$ satisfying
\begin{equation}
\left\{
\begin{split}
& \left(- \frac{d^2}{dz^2} + q_{{\bf e}}(z) - \lambda\right)\phi_{[{\bf e}]}(z,\lambda) = 0, \quad {\rm for} \quad  0 < z < 1,\\
& \phi_{[{\bf e}]}(0,\lambda) = 0, \quad \phi_{[{\bf e}]}'(0,\lambda) = 1.
\end{split}
\right.
\nonumber
\end{equation}
By the assumption (Q-3), $\phi_{[{\bf e}]}(z,\lambda)$ is determined by ${\bf e}$ and  independent of the orientation of ${\bf e}$. 
Then, the equation (\ref{S7equationbypsiwv}) is rewritten as
\begin{equation}
\widehat u(b) = - \frac{\phi_{[{\bf e}]}(1,\lambda)}{\phi_{[{\bf e}']}(1,\lambda)}\widehat u(b').
\label{S7ubandub'}
\end{equation}
Let ${\bf e}_{k,1}, {\bf e}'_{k,1}, {\bf e}_{k,2}, {\bf e}'_{k,2}, \cdots$ be the series of edges just below $A_k$ starting from the vertex $\alpha_k$, and put
\begin{equation}
f_{k,m}(\lambda) = - \frac{\phi_{[{\bf e}_{k,m}]}(1,\lambda)}{\phi_{[{\bf e}'_{k,m}]}(1,\lambda)}.
\label{s7equationforfkm}
\end{equation}
Then, we obtain the following lemma.

\begin{lemma}
The solution $\widehat u$ in  Lemma \ref{Lemmaspecialboundarydata} satisfies
$$
\widehat u(\alpha_{k,\ell}) = f_{k,1}(\lambda)\cdots f_{k,\ell}(\lambda).
$$
\end{lemma}

\subsection{Reconstruction procedure}
We now prove Theorem \ref{Maintheorem1} by showing the reconstruction algorithm of the potential $q_{\bf e}(z)$.

\medskip
\noindent
{\it 1st step}. We first take a sufficiently large hexagonal parallelogram $\Omega$ as in Figure \ref{S6HexaParallel} which contains all the supports of the potential $q_{\bf e}(z)$.

\medskip
\noindent
{\it 2nd step}. For an arbitrary $k$, draw a line $A_k$ as in  Figure \ref{LineAk} and take the boundary data $\widehat f$ having the properties in Lemma \ref{Lemmaspecialboundarydata}. 

\medskip
\noindent
{\it 3rd step}. Compute the values of the associated solution $\widehat u$ to the boundary value problem in Lemma \ref{Lemmaspecialboundarydata} at the points $\alpha_{k,\ell}$, $\ell = 0, 1, 2, \cdots$.

\medskip
\noindent
{\it 4th step}. Look at Figure \ref{S6HexaParallel}. Two edges ${\bf e}$ and ${\bf e'}$ between $A_k$ and $A_k'$ are said to be $A_k'$-adjacent if they have a vertex in common on $A_k'$ (see Figure \ref{S7hex_edge}).
 Take two $A_k'$-adjacent edges ${\bf e}$ and ${\bf e}'$ between $A_k$ and $A_k'$, and use the formula (\ref{s7equationforfkm}) to compute the ratio of $\phi_{[{\bf e}]}(1,\lambda)$ and $\phi_{[{\bf e}']}(1,\lambda)$.

\medskip
\noindent
{\it 5th step}. Rotate the whole system by the angle $\pi$ and take a hexagonal parallelogram congruent to the previous one. Then, the roles of $A_k$ and $A_k'$ are exchanged. One can then compute the  ratio of 
$\phi_{[{\bf e}]}(1,\lambda)$ and $\phi_{[{\bf e}']}(1,\lambda)$ for $A_k'$-adjacent pairs in the sense after the rotation, which are $A_k$-adjacent before the rotation.

\medskip
After the 4th and 5th steps, for all pairs ${\bf e}$ and ${\bf e}'$ which are either $A_k$-adjacent or $A_k'$-adjacent, one has computed the ratio of 
$\phi_{[{\bf e}]}(1,\lambda)$ and $\phi_{[{\bf e}']}(1,\lambda)$.

\medskip
\noindent
{\it 6th step}. Take a zigzag line on the hexagonal lattice (see Figure \ref{fig:line2}), and take any two edges ${\bf e}$ and ${\bf e}'$ on it. They are between $A_k$ and $A_k'$ for some $k$. Then, using the 4th and 5th steps, one can compute the ratio of $\phi_{[{\bf e}]}(1,\lambda)$ and $\phi_{[{\bf e}']}(1,\lambda)$ by computing the ratio for two successive edges between ${\bf e}$ and ${\bf e}'$.

\begin{figure}[h]
  \begin{center}
   \includegraphics[width=60mm, bb=0 0 633 440]{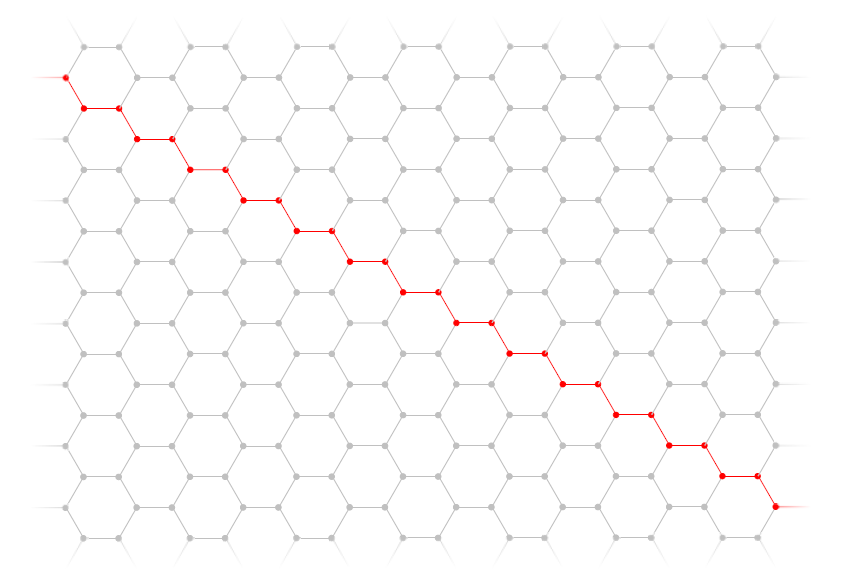}
  \end{center}
  \caption{Zigzag line in the hexagonal lattice}
  \label{fig:line2}
\end{figure}

\medskip
\noindent
{\it 7th step}.  For a sufficiently remote edge ${\bf e}'$, one knows $\phi_{[{\bf e}']}(1,\lambda)$ since $q_{{\bf e}'}(z) = 0$ on ${\bf e}'$. One can thus compute $\phi_{[{\bf e}]}(1,\lambda)$ for any edge ${\bf e}$.
Then, by the analytic continuation, one can compute the zeros of $\phi_{[{\bf e}]}(1,\lambda)$ for any edge ${\bf e}$. 

\medskip
\noindent
{\it 8th step}. Note that the zeros of  $\phi_{[{\bf e}]}(1,\lambda)$ are the Dirichlet eigenvalues for the operator $- (d/dz)^2 + q_{\bf e}(z)$ on $(0,1)$. Since the potential is symmetric, by Borg's theorem (see e.g. \cite{PoTru}, p. 117) these eigenvalues determine the potential $q_{\bf e}(z)$. 

\medskip
We have now completed the proof of Theorem \ref{Maintheorem1}.

\medskip
Note that for the 1st step, we need a-priori knowledge of the size of the support of the potential $q_{\mathcal E}(z)$. The knowledge of the D-N map is used in the 2nd step (in the proof of  Lemma \ref{S6partialDNdata}). In the 3rd step, one uses the equation (\ref{S7BVP}) and the fact that $\widehat u = 0$ below $A_k$.

\medskip
The proof of Theorem \ref{Maintheorem2} requires no essential change. Instead of $\frac{\sin\sqrt{\lambda}z}{\sqrt{\lambda}}$ and $\frac{\sin\sqrt{\lambda}(1-z)}{\sqrt{\lambda}}$, we have only to use the corresponding solutions to the Schr{\"o}dinger equation 
$\big(- (d/dz)^2 + q_0(z) - \lambda)\varphi = 0$.

\bigskip
\noindent
{\bf Acknowledgement}
K. A. is supported by Grant-in-Aid for Scientific Reserach (C) 17K05303, Japan Society for the Promotion of Science (JSPS).
H. I.  is supported by Grant-in-Aid for Scientific Research (C) 20K03667, JSPS. 
 E. K. is supported by the RFBR grant No. 19-01-00094. 
H. M. is supported by Grant-in-Aid for Young Scientists (B) 16K17630, JSPS. 
The authors express their gratitude to these supports.


\end{document}